%% file: main.tex
\documentclass[11pt]{article}

\input{header}
\usepackage{amssymb}

\usepackage{amsfonts}

\begin{document}

%\pagenumbering{gobble}
%\clearpage
%\thispagestyle{empty}
%% Title
\input{title}

%% Abstract 
\input{abstract}

%\clearpage
%\pagenumbering{arabic}
%\newpage

%% Introduction
\input{intro}

\input{proof-overview}

%% Single batch analysis
\input{single_batch}

%% Proof overview + including notation
% \input{proof-overview}
\input{quantile-bounds}
\input{ub_many_batches}

\input{lb_many_batches}
\input{lb-4-batches}

%%% Acknowledgements
\section*{Acknowledgements}
We would like to thank Karin Strauss for helpful discussions about this work.

%%% References
\bibliographystyle{abbrv}
\bibliography{refs}

%%% Appendix
\appendix

%% Standard Inequalities
\input{std_inequalities}

\input{quantile-bounds-apx}

\end{document}

%% file: header.tex
\usepackage[utf8]{inputenc}
\usepackage{fullpage}
\usepackage[margin=1in]{geometry}

\usepackage{amsmath, amssymb,amsbsy,amsthm,dsfont}
\usepackage{graphicx}
\usepackage[dvipsnames]{xcolor}
\usepackage{caption}
\usepackage{subcaption}
\usepackage{enumerate}
\usepackage{nicefrac}
\usepackage{cases}
\usepackage{hyperref}%
\definecolor{bluish-green}{HTML}{009E73}
\hypersetup{%
   breaklinks,%
   colorlinks=true,%
   urlcolor=[rgb]{0.0,0.2,0.75},%
   linkcolor=[rgb]{0.5,0.0,0.0},%
   citecolor=bluish-green,%
   filecolor=[rgb]{0,0,0.4}, anchorcolor=[rgb]={0.0,0.1,0.2}%
}
\usepackage{mathtools}

\let\emptyset\varnothing

\allowdisplaybreaks[1]

\numberwithin{equation}{section}

\DeclarePairedDelimiter\ceil{\lceil}{\rceil}
\DeclarePairedDelimiter\floor{\lfloor}{\rfloor}

\DeclareMathOperator{\R}{\mathbb{R}} % Real numbers
\newcommand{\p}{\mathbb{P}} % Probability P
\newcommand{\E}{\mathbb{E}} % Expectation E
\newcommand{\eps}{\varepsilon} % Epsilon
\newcommand{\EQ}{\widetilde{Q}}
\newcommand{\nth}[1]{{#1}^\mathrm{th}}

\def\cB{{\mathcal B}}

\def\cE{{\mathcal E}}

\def\cL{{\mathcal L}}

\def\cS{{\mathcal S}}

\def\cU{{\mathcal U}}

\newcommand{\calB}{\mathcal{B}}
\newcommand{\calS}{\mathcal{S}}
\newcommand{\calR}{\mathcal{R}}

\newcommand{\cost}{\mathsf{cost}}
\newcommand{\ACGT}{\{A,C,G,T\}}
\newtheorem{theorem}{Theorem}[section]
\newtheorem{lemma}[theorem]{Lemma}

\newtheorem{claim}[theorem]{Claim}

\newtheorem{example}[theorem]{Example}
\newtheorem{remark}[theorem]{Remark}

\DeclareFontFamily{U}{mathx}{\hyphenchar\font45}
\DeclareFontShape{U}{mathx}{m}{n}{
      <5> <6> <7> <8> <9> <10>
      <10.95> <12> <14.4> <17.28> <20.74> <24.88>
      mathx10
      }{}
\DeclareSymbolFont{mathx}{U}{mathx}{m}{n}
\DeclareFontSubstitution{U}{mathx}{m}{n}
\DeclareMathAccent{\widebar}{0}{mathx}{"73}

%% file: title.tex
\title{Batch Optimization for DNA Synthesis}
\author{
	Konstantin Makarychev\thanks{Northwestern University; \url{konstantin@northwestern.edu}.}
	\and
	Mikl\'os Z.\ R\'acz\thanks{Princeton University; \url{mracz@princeton.edu}. Research supported in part by NSF grant DMS 1811724 and by a Princeton SEAS Innovation Award.}
	\and
	Cyrus Rashtchian\thanks{University of California, San Diego; \url{crashtchian@eng.ucsd.edu}.}
	\and
	Sergey Yekhanin\thanks{Microsoft Research; \url{yekhanin@microsoft.com}.}
}
\date{\today}

\maketitle

%% file: abstract.tex
\begin{abstract}
Large pools of synthetic DNA molecules have been recently used to reliably store significant volumes of digital data. While DNA as a storage medium has enormous potential because of its high storage density, its practical use is currently severely limited because of the high cost and low throughput of available DNA synthesis technologies. 

We study the role of batch optimization in reducing the cost of large scale DNA synthesis, which translates to the following algorithmic task. Given a large pool $\mathcal S$ of random quaternary strings of fixed length, partition $\mathcal S$ into batches in a way that minimizes the sum of the lengths of the shortest common supersequences across batches.  

We introduce two ideas for batch optimization that both improve (in different ways) upon a naive baseline: (1) using both $(ACGT)^{*}$ and its reverse $(TGCA)^{*}$ as reference strands, and batching appropriately, and (2) batching via the quantiles of an appropriate ordering of the strands. We also prove asymptotically matching lower bounds on the cost of DNA synthesis, showing that one cannot improve upon these two ideas. Our results uncover a surprising separation between two cases that naturally arise in the context of DNA data storage: the asymptotic cost savings of batch optimization are significantly greater in the case where strings in $\mathcal S$ do not contain repeats of the same character (homopolymers), as compared to the case where strings in $\mathcal S$ are unconstrained. 
\end{abstract}

%% file: intro.tex
%\newpage
\section{Introduction} \label{sec:intro} 

Storing digital data in synthetic DNA molecules has received much attention in the past decade~\cite{Neiman1965,church2012next,goldman, yazdi2015dna,bornholt2016dna,shipman2017crispr,yazdi2017portable,erlich2017dna,organick2018random,ceze2019molecular,meiser2020reading}. DNA data storage offers several orders of magnitude higher information density compared to conventional storage media, as well as the potential to store data reliably for hundreds or thousands of years. However, the prohibitively high cost and low throughput of modern DNA synthesis technologies present a key barrier that needs to be addressed in order to make DNA data storage a commonplace technology. 

For the purposes of the current paper we can think of a DNA molecule as a string (strand) in the quaternary alphabet $\{A,C,G,T\}.$ Today the dominant method for producing large quantities of DNA molecules is array-based DNA synthesis~\cite{KosuriChurch2014,Howon2020}.  With this technology the DNA synthesizer creates a large number of DNA strands in parallel, where each strand is grown by one DNA base (character) at a time. To append bases to strands, the synthesis machine follows a fixed supersequence of bases, called a reference strand. As the machine iterates through this supersequence, the next base is added to a select subset of the DNA strands. This process continues until the machine reaches the end of the supersequence. In particular, each synthesized DNA strand must be a subsequence of the reference strand. The cost of DNA synthesis is proportional to the length of the reference strand.

In applications to DNA data storage one typically needs to synthesize very large quantities of DNA molecules, significantly exceeding the capacity of any single DNA synthesizer. Therefore the pool of strands that one aims to synthesize needs to be partitioned into batches, where the size of each batch corresponds to the maximum load of the synthesizer. In this setting the total cost of DNA synthesis is proportional to the sum of the lengths of the shortest common supersequences of each batch. The focus of this paper is the algorithmic task of \emph{batch optimization}, where the goal is to partition the strands into batches and assign every batch a reference strand in a way that minimizes this cost. 

The encoding process that generates the list of DNA strands that need to be synthesized to store a given digital file varies with the specific system~\cite{church2012next,bornholt2016dna,organick2018random,ICASSP2020} and is usually quite complex. The encoder adds redundancy to the data to allow for the correction of various types of errors that occur during DNA synthesis, storage, and sequencing, including insertions, deletions, and substitutions of individual bases, as well as missing DNA strands. 

We now describe two aspects of encoding of digital data in DNA that are relevant to our work. Commonly, input digital data is  randomized~\cite{organick2018random} using a seeded pseudorandom number generator or compressed and encrypted~\cite{ICASSP2020}; this is done in order to reduce the frequency of undesirable patterns that may occur in strands that are used to represent the data, for instance, patterns likely to cause the presence of DNA secondary structure~\cite{Bochman2012}. Ensuring that strands look random also facilitates certain tasks that may be a part of the decoding process such as clustering~\cite{organick2018random,rashtchian2017clustering} and trace reconstruction~\cite{BatuKannan04-RandomCase,HolensteinMPW08,viswanathan,organick2018random,peres2017average,holden2018subpolynomial}. Another important aspect is as follows. Algorithms that encode digital data in DNA~\cite{goldman,organick2018random} often ensure that the resulting strands do not contain long runs of the same character (i.e., homopolymers), since such runs are known to cause errors during the DNA sequencing stage. The length of the longest allowed homopolymer run may be as low as~one---that is, not allowing homopolymers---or unconstrained, depending on the scenario.

Motivated by the above considerations, in the current paper we model pools $\mathcal S$ of DNA strands that we aim to synthesise as large collections of random quaternary strings. We consider two key representative cases: the case where strings in~$\mathcal S$ are unconstrained and the case where strings in~$\mathcal S$ do not contain repeats of the same character.

\subsection{Problem statement}\label{sec:problem_statement}

Fix a strand length $n,$ and consider two different choices for the strand universe $\mathcal U.$ 
\begin{enumerate}
    \item \emph{Unconstrained strands:} $\mathcal U = \{A,C,G,T\}^n.$ 
    \item \emph{Strands without homopolymers:} $\mathcal U$ is the subset of $\{A,C,G,T\}^n$ that contains all strands with no consecutively repeated characters.
\end{enumerate}
Let $\mathcal S$ be a subset of elements of $\mathcal U,$ with $M := |\mathcal S|$; this is the pool of strands we wish to synthesize. Let $k$ be an integer that divides $M,$ and let $\pi$ be a partition of $\mathcal S$ into $k$ subsets (which we refer to as \emph{batches}) $\mathcal B_1,\ldots, \mathcal B_k$ of size $M/k$.\footnote{The assumption that the batches are of equal size is made for simplicity. Indeed, our techniques extend to a more general setting where the batches are roughly the same size (e.g., up to a constant factor), and several results are phrased in this more general setting.} We define $\cost(\mathcal B_i)$, the cost of synthesizing elements of the batch $\mathcal B_i$, as the length of the shortest common supersequence of all strands in $\mathcal B_i.$ Using this notation, we define the cost of synthesizing the whole pool $\mathcal S$ as:
\begin{equation}
    \cost(\mathcal S) := \min\limits_{\pi}  \sum\limits_{i=1}^k\cost(\mathcal B_i).
\end{equation}
We assume that elements of $\mathcal S$ are selected i.i.d.\ from $\mathcal U$ uniformly at random, and we are interested in upper and lower bounds for $\cost(\cS)$ that hold with high probability.

While in the practice of DNA synthesis the parameters $n$, $M$, and $k$ are concrete numbers, to facilitate the asymptotic study of the problem we focus on the following relevant scenario: $n$ is growing, $M$ is significantly larger than but polynomial in $n,$ and $k$ is either a constant or a slowly growing function of $n.$ 

\begin{example}\label{ex:illustration}
Consider the setting of strands with no homopolymers. Let $n=4$ and $M=4.$ Let 
$\mathcal S=\{AGCT,GCAT,CAGA,GAGC\}.$ Assume that $k=2,$ that is, there are two batches and each batch contains two strands. 

We can partition $\cal S$ into $\mathcal B_1=\{AGCT, GCAT\}$ and $\mathcal B_2=\{CAGA, GAGC\}$. The DNA synthesizer (printer) first prints $\calB_1$. It starts with two empty strings $(\emptyset,\emptyset)$. Then, it appends $A$ to the first strand and obtains strands $(A,\emptyset)$. It appends $G$ to both strands and obtains $(AG,G)$. Then, it appends the letters $C$, $A$, and $T$ as follows:
$$(\emptyset, \emptyset) \xrightarrow{A} (A, \emptyset)\xrightarrow{G} (AG, G)\xrightarrow{C} (AGC, GC)\xrightarrow{A} (AGC, GCA)\xrightarrow{T} (AGCT, GCAT).$$
After the last step, we get the set $\calB_1 = \{ AGCT, GCAT\}$. The printer prints $\calB_2$ as follows:
$$(\emptyset, \emptyset) \xrightarrow{C} (C, \emptyset)\xrightarrow{G} (C, G)\xrightarrow{A} (CA, GA)\xrightarrow{G} (CAG, GAG)\xrightarrow{A} (CAGA, GAG)\xrightarrow{C} (CAGA, GAGC).$$
In this example, we used the reference strand $AGCAT$ to print the set $\calB_1$ in five steps and the reference strand $CGAGAC$ to print the set $\calB_2$ in six steps. Therefore $\cost(\mathcal S)\leq 11.$
\end{example}

%%%%%%%%%%%%%%%%%%%%%%%%%%%%%%%%%%%%%%%%%%%%%%%%%%
\subsection{Main results for multiple batches}\label{sec:results} %%%
%%%%%%%%%%%%%%%%%%%%%%%%%%%%%%%%%%%%%%%%%%%%%%%%%%
Before describing our main results for multiple batches, 
we briefly and informally discuss the setting of a single batch---formal statements and proofs are in Section~\ref{sec:single_batch}. 
A natural reference strand to use to print a pool of strands $\cS$ is the periodic strand $(ACGT)^{*}$, where $ACGT$ repeats indefinitely. 

For this reference strand, the cost of printing a random strand can be written as 
$\sum_{i=1}^{n} X_{i}$, 
where $\left\{ X_{i} \right\}_{i=1}^{n}$ 
are i.i.d.\ uniformly random on $\{1,2,3,4\}$ in the case of unconstrained strands; 
in the case of strands without homopolymers, $\left\{ X_{i} \right\}_{i=1}^{n}$ are independent, with $X_{1}$ uniformly random on $\{1,2,3,4\}$ and $X_{i}$ uniformly random on $\{1,2,3\}$ for $i \geq 2$. 
By using a standard concentration inequality we then obtain the upper bounds 
$\cost \left( \cS \right) \leq 2.5 n + 3 \sqrt{n \log M}$ 
for unconstrained strands and 
$\cost \left( \cS \right) \leq 2 n + 3 \sqrt{n \log M}$ 
for strands without homopolymers, 
with both bounds holding with probability $1-o(1)$. 
Combining this with an appropriate stochastic domination argument that compares random walks, we also obtain matching lower bounds, for both choices of the strand universe $\cU$. 
This shows that for a single batch no reference strand can do asymptotically better than the periodic strand $(ACGT)^{*}$.

The setting of multiple batches, which is the focus of our work, presents interesting challenges. As a simple baseline, we could consider randomly partitioning $\cS$ into $k$ batches. A direct application of the single batch upper bound would provide a cost of roughly $2.5nk + O(k \sqrt{n \log(M/k)})$ for unconstrained strands and $2nk + O(k \sqrt{n \log (M/k)})$ for strands without homopolymers. We provide improvements in both cases by using a slightly more sophisticated batching method.

%%The more substantial improvement arises for strands without homopolymers.

We first observe a symmetry property: For any strand without homopolymers the cost of printing it using $(ACGT)^{*}$ and the cost of printing it using $(TGCA)^{*}$ add up to $4n+1$, so the better choice of reference strand results in a cost of at most $2n$. This idea can be extended to a large set of strands, by choosing for each strand the better reference strand out of $(ACGT)^{*}$ and its reverse $(TGCA)^{*}$. 
We further improve upon the cost by leveraging a second idea. 
After partitioning strands based on which of the two reference strands is better, we then sort the strands based on their cost (with respect to the chosen reference strand). We then use a quantile-based batching process to group the first $M/k$ lowest cost strands, then the next $M/k$, etc. 
We show that combining these two ideas reduces the total cost to $2nk - \Theta(k\sqrt{n})$ for $k \geq 3$ batches.

In the case of unrestricted strands, such an improvement is not possible, although we are able to show that with $k$ batches a similar partitioning strategy, based on appropriately ordering the strands and using quantiles, enables us to save a factor of $k$ in the deviation term and obtain a total cost of $2.5nk + O(\sqrt{n \log M})$. 
We now formally state our results.

\begin{theorem}[Upper bounds]\label{thm:ub-multi}
		Let $\calS$ be a set of $M$ random strands in  $\ACGT^n$, and let $k$ be an integer satisfying $3 \leq k \leq \frac{1}{4}\sqrt{\frac{M}{\log M}}$. There exist absolute constants $C_{1} > 0$ and $C_{2} < \infty$ such that the following hold.
		\begin{enumerate}		
		\item \emph{\textbf{(Strands without homopolymers)}} 
		There exists a way to efficiently partition $\calS$ into $k$ equal size batches $\calB_1, \ldots, \calB_k$ such that with probability at least $1-1/M$ we have that 
% 		$|\calB_i| \in \frac{M}{k} \pm \Theta(\sqrt{M \log M}/k)$ and 
		\[
		\sum_{i=1}^k \cost(\calB_i) \leq 2nk - C_1 k \sqrt{n}.
		\]
		\item \emph{\textbf{(Unconstrained strands)}} 
		There exists a way to efficiently partition $\calS$ into $k$ equal size batches $\calB_1, \ldots, \calB_k$ such that with probability at least  $1-1/M$ we have that
		\[
		\sum_{i=1}^k \cost(\calB_i) \leq 2.5nk + C_2 \sqrt{n \log M}.
		\]
	\end{enumerate}
\end{theorem}
We complement these results with almost tight lower bounds.
%%With batching there is potentially an even greater possibility of benefiting from variable reference strands that may be tailored to the individual batches. Despite this flexibility, we show that we can not significantly improve Theorem~\ref{thm:ub-multi} for either class of strings, at least when $k$ is not too large. 
Proving the following theorem is the most technically challenging part of our work.

\begin{theorem}[Lower bounds]\label{thm:lb-4-homoplymers-main}
Let $\calS$ be a set of $M \geq 10 n^2 \log n$ random strands in $\ACGT^n$, and let $k$ be a positive integer satisfying $k \leq \frac{1}{10}\sqrt{\log M / \log \log M}.$ 
	\begin{enumerate}		
	\item \emph{\textbf{(Strands without homopolymers)}} There exists an absolute constant $c_1 < \infty$ such that the following holds with probability at least $1 - c_{1}/M$. For any  partition of~$\calS$ into $k$ equal size batches $\calB_1, \ldots, \calB_k$, we have that
	\[
	\sum_{i=1}^k \cost(\calB_i) \geq 2nk - c_1k\sqrt{n\log k}.
	\]
	\item \emph{\textbf{(Unconstrained strands)}} Suppose that $M \leq \exp(n)$. There exists an absolute constant $c_2 >0$ such that the following holds with probability at least $1 - c_{2}^{-1}/M$. For any  partition of~$\calS$ into $k$ equal size batches $\calB_1, \ldots, \calB_k$, we have that
	\[
	\sum_{i=1}^k \cost(\calB_i) \geq 2.5nk + c_2\sqrt{n \log M}.
	\]
\end{enumerate}
\end{theorem}

Comparing Theorems~\ref{thm:ub-multi} and~\ref{thm:lb-4-homoplymers-main}, we see that the upper and lower bounds match up to the absolute constants in the deviation term when $k$ is small enough. As a consequence, this provides evidence that our 
%greedy 
batching method is nearly optimal, perhaps surprisingly. 

Furthermore, Theorems~\ref{thm:ub-multi} and~\ref{thm:lb-4-homoplymers-main} provide a clear separation between the two representative strand universes. 
On the one hand, for unconstrained strands we have, with probability $1-o(1)$, that 
$\cost(\cS) = 2.5 n k + \Theta \left( \sqrt{n \log M} \right)$;  
that is, the cost \emph{exceeds} the main term $2.5 n k$ by the deviation term. 
On the other hand, for strands without homopolymers we have, with probability $1-o(1)$, that 
$2nk - c_{1} k \sqrt{n \log k} 
\leq \cost(\cS) 
\leq 2nk - C_{1} k \sqrt{n}$; 
that is, the cost is \emph{smaller than} the main term $2 nk$ by the deviation term. 

\subsection{Related work}

For an overview of the biochemical DNA synthesis process, we refer the interested reader to the surveys~\cite{KosuriChurch2014,ceze2019molecular}. Our work is motivated by several experimental papers that address the challenge of reducing the synthesis cost in both single and multi-batch settings~\cite{hannenhalli2002combinatorial, kahng2002border, rahmann2003shortest, kahng2004scalable, ning2006distribution, rahmann2006subsequence, kumar2010dna, trinca2011parallel,ning2011multiple, srinivasan2014review}. Variants of the problem have also been studied that incorporate certain quality control measures~\cite{hubbell1999fidelity,colbourn2002construction,sengupta2002quality,milenkovic2006error}. Much of this previous work considers the  $(ACGT)^*$ supersequence when analyzing the synthesis cost. Rahmann first observed that in this case the single batch cost of uniformly random strings is approximately Gaussian, but he did not provide a formal analysis nor any asymptotic or finite-size bounds~\cite{rahmann2003shortest}. 
%(see also~\cite{pemantle2008twenty} for a generating function argument)
%flaxman2004strings, L66, pemantle2004asymptotics, pemantle2008twenty
In the multi-batch setting, previous work uses the same cost function as we do, namely the sum of the shortest common supersequence (SCS) lengths for each batch~\cite{ning2006distribution,ning2011multiple}. In general, a wide array of algorithms have been proposed and empirically evaluated for selecting a short reference string given the set of DNA strands to synthesize. However, these heuristics do not come with provable guarantees, and many of them implicitly solve the SCS problem, which is known to be NP-hard for a collection of strings~\cite{raiha1981shortest,JiangLi1995}.  

From a theoretical point of view, a few recent works have considered minimizing the synthesis cost through coding-based approaches. Lenz et al.~study reference strings that have a large number of subsequences, and they consider mappings to encode data by a set of strings while minimizing the single-batch synthesis cost~\cite{lenz2020coding}.   A slightly different synthesis model has also been considered, where information is stored based on run-length patterns in the strings~\cite{anavy2019data,lee2019terminator,jain2020coding}. Our work also relates to combinatorial questions about the number of distinct subsequences~\cite{flaxman2004strings, hirschberg_tight_2000, L66,  pemantle2004asymptotics, pemantle2008twenty}.

There is also a large body of prior work on the longest common subsequence (LCS) of random strings~\cite{bukh2019periodic,chvatal1975longest,danvcik1995upper,houdre2016closeness,kiwi2005expected,lueker2009improved,navarro2001guided}. The expected LCS length of two random length $n$ strings is known to be $(\gamma+o(1)) n$ for a value $\gamma >0$ called the Chv\'{a}tal-Sankoff constant. Despite decades of effort, the exact value of $\gamma$ remains unknown for constant alphabet sizes. For two length $n$ strings, the LCS and SCS are related via the equality $\mathsf{SCS}(S_1,S_2) = 2n - \mathsf{LCS}(S_1,S_2)$, but for larger sets, no analogous relationship is known. In particular, our results show that the average SCS length for a large collection of strings behaves very differently than for a pair of strings. While we are not aware of prior results on the SCS for multiple batches, our single batch results improve an existing bound on the expected SCS length in the special case of $M=n$ strings (see Remark~\ref{rem:SCS}).

\subsection{Organization}

The rest of the paper is organized as follows. We begin in Section~\ref{sec:proof-overview} with a technical overview of our proofs. In Section~\ref{sec:single_batch}, we provide both upper and lower bounds for a single batch. In Section~\ref{sec:emp-quant}, we introduce the cost quantile preliminaries that we use for our multi-batch results. We prove the upper bounds for the multi-batch setting, Theorem~\ref{thm:ub-multi}, in Section~\ref{sec:ub_multi_batch}. Finally, we prove the lower bounds for the multi-batch setting, Theorem~\ref{thm:lb-4-homoplymers-main}, in Section~\ref{sec:lb_multi_batch}.

%% file: proof-overview.tex
\section{Proof Overview}\label{sec:proof-overview}

In this section we give an overview of our results and the associated proofs. Suppose we want to synthesize a DNA strand $S$ using a reference strand $R$. Denote the length of the prefix of $R$ which we use for synthesis by $\cost_R(S)$. Then, the cost of printing a batch of strands $\calB$ using $R$ equals the maximum cost of printing $S$ for $S\in \calB$:
$$\cost_R(\calB) = \max_{S\in \calB} \cost_R(S).$$
We observe that the cost of printing any strand of length $n$ using the periodic reference strand $(ACGT)^*$ is at most $4n$, since the $i$-th base of $S$ can be printed using the corresponding base in the $i$-th quadruple of $(ACGT)^*$. Hence, the cost of synthesizing any batch of strands of length $n$ is bounded from above by $4n$. As we discuss later, the cost of every strand without homopolymers with respect to the reference strand $(ACGT)^*$ is at most $3n+1$. So the cost of any batch of strands without homopolymers is also at most $3n+1$. 

Since the cost of synthesizing every batch of strands is upper bounded by $4n$, we do not need to consider reference strands of length more than $4n$. However, for the sake of analysis, we shall assume that all reference strands $R$ have an infinite length. The first $4n$ bases of these strands are arbitrary, while the remaining infinite suffix is a repetition of the pattern $ACGT$. We denote the set of all such strands by $\calR^*$. Observe that every strand $\calS$ can be synthesized using every $R\in \calR^*$ because $R$ contains the substring $(ACGT)^*$. Note that when we synthesize a batch $\calB$ using a reference strand $R\in \calR^*$, we truncate $R$ after $\cost_R(\calB)$ bases, so effectively we use a reference strand of length $\cost_R(\calB)$.

\subsection{Cost of a Single Batch}

We first show how to estimate the cost of synthesizing a single batch of DNA strands. We prove that for a random strand $S$ of length $n$ and reference strand $\widetilde R=(ACGT)^*$, the expected $\cost_{\widetilde R}(S)$ equals $2.5 n$. We then use concentration inequalities to argue that the maximum cost of strands in $\calB$ is upper bounded by $2.5 n + O(\sqrt{n\log M})$ with high probability, where $M$ is the batch size. Similarly, we show that for every fixed strand $R$, we have that $\E[\cost_{R}(S)] \geq 2.5 n$. Hence, for every fixed $R$ the cost of $\calB$ is also lower bounded by $2.5 n + \Omega(\sqrt{n\log M})$ with high probability. We obtain a lower bound on the cost of a batch by taking the union bound over all $R\in \calR^*$. Similarly, we get lower and upper bounds of  $2 n + \Omega(\sqrt{n\log M})$ and $2 n + O(\sqrt{n\log M})$ for random strands without homopolymers. 

We now discuss how to compute $\E[\cost_{R}(S)]$ for a given reference strand $R$ and random $S$. Let $\tau_i(S,R)$ be the cost of the prefix $S_1,\dots, S_i$. In other words, $\tau_i(S,R)$ is the index of the base in $R$ that is used for synthesizing the $i$-th base in $S$. We let $\tau_0(S,R)=0$. Observe that $\left\{ \tau_i(S,R) \right\}_{i \geq 0}$ is a Markov chain: the value of $\tau_{i+1}(S,R)$ depends only on the current state $\tau_{i}(S,R)$ and the random value of $S_{i+1}$. We denote the increments of $\tau_{i}(S,R)$ by $X_i(S,R)$: for $i\in\{1,\dots, n\}$, let
$$X_i(S,R) := \tau_{i}(S,R) - \tau_{i-1}(S,R).$$
Then, $\cost_R(S) = \tau_{n}(S,R) = \sum_{i=1}^n X_i(S,R)$. For the reference strand $\widetilde R=(ACGT)^*$, each increment $X_i(S,\widetilde R)$ is a random variable uniformly distributed in $\{1,2,3,4\}$, and all $X_i(S,\widetilde R)$ are mutually independent. Consequently, 
$\E[X_i(S,\widetilde R)] = 2.5$ for all $i$ and thus  $\E[\cost_{\widetilde{R}}(S)] = 2.5 n$. Furthermore, by the central limit theorem, the deviation of the cost from its expectation, $\cost_{\widetilde{R}}(S)-2.5n$, is approximately Gaussian with mean $0$ and variance $1.25 n$. Thus, we can use Hoeffding's inequality and other concentration inequalities to obtain upper and lower bounds of on $\cost_{\widetilde R}(S)$.  These bounds imply that the cost of a single batch of $M$ strands equals $2.5n+\Theta(\sqrt{n \log M})$.

To show that $\E[X_i(S,R)] \geq 2.5$ for every $R\in \calR^*$ and not only for $R= \widetilde R$, we observe that the sequence $X_1(S,R),\dots,X_n(S, R)$ \emph{stochastically dominates} a sequence of i.i.d random variables $Y_1,\dots, Y_n$, where each $Y_i$ is uniformly distributed in $\{1,2,3,4\}$. Hence, 
$$\E[X_1(S,R)+\cdots +X_n(S,R)]
\geq  
\E[Y_1+\cdots +Y_n] = 2.5n.$$

For random strands without homopolymers, each jump $X_i (S,\widetilde R)$ is uniformly distributed in $\{1,2,3\}$ for $i>1$; and $X_1 (S,\widetilde R)$ is uniformly distributed in $\{1,2,3,4\}$. Hence, the expected cost $\cost_{\widetilde R}(S)$ is $2n + \nicefrac{1}{2}$. Also, note that the maximum possible value of $X_i (S,\widetilde R)$ is $3$ (for $i>1$). Hence, the cost of every strand is upper bounded by $3n+1$.

\subsection{Upper Bounds for Multiple Batches}
We are going to use the same reference strand $\widetilde R = (ACGT)^*$ for synthesizing all batches of unconstrained strands and two different reference strands, $\widetilde R = (ACGT)^*$
and its reverse $\widebar R = (TGCA)^*$, for synthesizing batches of strands without homopolymers. 

\medskip
\noindent\textbf{Na\"ive Approach.} Suppose we assign strands randomly to $k$ batches. Then, each batch consists of $M/k$ random strands sampled uniformly from $\{A,C,G,T\}^n$. Hence, the cost of every batch is $2.5n+\Theta(\sqrt{n \log M})$. Consequently, the total cost of synthesising $k$ batches is $2.5nk+k\cdot\Theta(\sqrt{n \log M})$. We now show that by carefully assigning strands to batches we can improve this cost to 
$2.5nk+\Theta(\sqrt{n \log M})$ for unconstrained strands. Similarly, we show how to improve a na\"ive solution of cost $2nk+k\cdot\Theta(\sqrt{n \log M})$ for strands without homopolymers to a solution of cost $2nk-\Omega(k\sqrt{n})$.

\medskip

\noindent\textbf{Unconstrained Strands.} 
Our strategy for splitting the set of unconstrained strands $\calS$ into $k$ batches is quite simple. For every strand $S$ in $\calS$, we compute $\cost_{\widetilde{R}}(S)$ and then sort strands by this cost. We put the first $M/k$ strands in the first batch, the second $M/k$ strands in the second batch, and so on. Then, the cost of the $i$-th batch is equal to the empirical $\nicefrac{i}{k}$-th quantile of $\left\{ \cost_{\widetilde{R}}(S) \right\}_{S \in \calS}$. We formally define empirical quantiles in Section~\ref{sec:emp-quant}. In Section~\ref{sec:emp-quant}, we also show that, with high probability, empirical quantiles of $\left\{ \cost_{\widetilde{R}}(S) \right\}_{S \in \calS}$ are very close to the corresponding quantiles of the distribution of the random variable $\cost_{\widetilde R}(S)$, where $S$ is randomly and uniformly drawn from $\{A,C,G,T\}^n$. The only exception is the empirical $1$-quantile of the sample $\calS$ which corresponds to the cost of the most expensive strand in $\calS$.  This cost is approximately equal to the $(1-\nicefrac{1}{M})$-quantile of the distribution of  $\cost_{\widetilde R}(S)$, where $M$ is the size of $\calS$.

As we discussed above, $\cost_{\widetilde R}(S)$ can be approximated by the random variable $2.5n + g$, where $g$ is a Gaussian random variable with mean $0$ and variance $1.25 n$. The sum of the $1/k, 2/k, \dots, (k-1)/k$ quantiles of a symmetric Gaussian distribution equals $0$, since the quantiles $i/k$ and $(k-i)/k$ are symmetric around $0$. However, the $(1-\nicefrac{1}{M})$-quantile of the distribution of $g$ is relatively large and approximately equals $c \sqrt{n \log M}$. Hence, the total cost of synthesizing $k$ batches approximately equals 
$$2.5nk + c\sqrt{n \log M}.$$
We make this argument formal in Section~\ref{sec:ub_multi_batch}.

\medskip

\noindent\textbf{Strands without Homopolymers.} If we use the same batching strategy as we discussed above for strands without homopolymers, we obtain a solution of cost 
$2nk + c\sqrt{n \log M}$ with high probability. However, somewhat surprisingly, we can do better by utilizing two reference strands, $\widetilde{R}=(ACGT)^*$ and 
$\widebar{R}=(TGCA)^*$, instead of just the single strand $\widetilde R$. We show that the 
random variables $\cost_{\widetilde R} (S)$ and 
$\cost_{\widebar R} (S)$ are anticorrelated. Specifically, for every strand $S$ without homopolymers, we (deterministically) have 
$$\cost_{\widetilde R} (S) + \cost_{\widebar R} (S) = 4n+1.$$

% This observation suggests the following strategy: We first split all strands in $\calS$ into two groups: strands for which 
% $\cost_{\widetilde R}(S)\leq \cost_{\widebar R}(S)$
% and strands for which 
% $\cost_{\widebar R}(S) < \cost_{\widetilde R}(S)$. We then sort strands in each group by cost and partition each group into batches of size $\approx M/(2k)$ as we did before (see Section~\ref{sec:ub_multi_batch_no_homopoly} for details).

%{\color{blue}
This observation suggests the following strategy: We first sort all strands $\calS$ by their cost when printed with $\widetilde R$. For the first $\lceil k/2 \rceil$ batches, we print them with $\widetilde R$, and we print the remaining batches with $\widebar R$. Overall, we will argue that this batching process results in $k-2$ batches having a cost of at most $2n$, and a constant fraction of these batches having an additional savings of $\Omega(\sqrt{n})$, which results in the ultimate savings of  $\Omega(k \sqrt{n})$. 
The only challenging batches are the ``middle'' two. We handle these by arguing that their costs are coupled so that together they do not exceed $4n+1$. 
We next explain the intuition behind the main savings. See Section~\ref{sec:ub_multi_batch_no_homopoly} for full details of the batching process and the analysis.
%}

Since $(X_i(S,\widetilde R) + X_i(S,\widebar R))/2 = 2$ for all $i > 1$ and $S$ does not have homopolymers, the random variables $\cost_{\widetilde R} (S)$ and
$\cost_{\widebar R} (S)$ can be approximated by \emph{correlated} random variables $2n - g$ and $2n + g$, where $g$ is a Gaussian random variable with mean $0$ and variance $\nicefrac{2}{3}\,n$. The cost of every strand is thus approximately equal to $\min\{2n-g, 2n+g\} = 2n - |g|$, and the total cost of $k$ batches is approximately equal to 
the sum of the $\nicefrac{i}{k}$-quantiles of the random variable $2n - |g|$ for $i=1,\dots, k$. For sufficiently large $k$, this sum is approximately equal to 
$$k\cdot(\E[2n - |g|]) = k\cdot(2n - \E[|g|])= 2nk - k\sqrt{\frac{4}{3\pi}\,n}.$$
For small $k$ ($k>2$), the sum is upper bounded by $2nk -\Omega(k\sqrt{n})$. We prove this bound for 
%{\color{blue} 
$k\geq 3$
%} 
in  Section~\ref{sec:ub_multi_batch_no_homopoly}.

\subsection{Lower Bounds for Multiple Batches}
We now discuss how to obtain lower bounds on the cost of batch synthesis. We start with lower bounds that are based on the following observation: Every batch $\calB$ must contain a  $\nicefrac{1}{k}$ fraction of all strands in $\calS$. Consequently, its cost is lower bounded by the empirical $\nicefrac{1}{k}$-quantile of $\left\{ \cost_R(S) \right\}_{S \in \calS}$, which, in turn, approximately equals the  $\nicefrac{1}{k}$-quantile of the distribution of the random variable $\cost_R(S)$, where $S$ is a random strand. Here $R$ is the reference strand used for synthesising $\cB$. Using the notation (defined in Section~\ref{sec:emp-quant}) for empirical $q$-quantiles $\EQ_{q,R}(\calS)$ and $q$-quantiles $Q_{q,R}(D)$ of a distribution $D$, we can lower bound the cost of $\cB$ as follows:
$$\cost(\calB)\geq \min_{R\in \calR^*} \EQ_{\nicefrac{1}{k},R}(\calS) \gtrsim \min_{R\in \calR^*} Q_{\nicefrac{1}{k},R}(D_{\nicefrac{1}{4}}),$$
where $D_{\nicefrac{1}{4}}$ is the uniform distribution of strands of length $n$. Using Hoeffding's inequality for $\cost_R(S)$ along with bounds on $\EQ_{\nicefrac{1}{k},R}(\calS)$ and $Q_{\nicefrac{1}{k},R}(D_{\nicefrac{1}{4}})$ from Section~\ref{sec:emp-quant}, we then show that $Q_{\nicefrac{1}{k},R}(D_{\nicefrac{1}{4}})\geq 2.5n - O(\sqrt{n \log k})$ which yields a lower bound of $k\cdot (2.5n - O(\sqrt{n \log k}))$ on the total cost of synthesizing $k$ batches. For strands without homopolymers, the same argument gives a bound of $k\cdot (2n - O(\sqrt{n \log k}))$.

\medskip

\noindent\textbf{Improved Lower Bound for Unconstrained Strands.}
We then improve the lower bound on the cost of batch synthesis of  unconstrained strands by showing that while the cost of all batches are lower bounded by  $2.5n - O(\sqrt{n \log k})$, the cost of the most expensive batch is at least $2.5n + \Omega(\sqrt{n \log M})$. Note that a similar statement does not hold for strands without homopolymers. To prove that the cost of the most expensive batch is 
$2.5n + \Omega(\sqrt{n \log M})$, we consider a subset $\calS''$ of strands that have disproportionately many (roughly, 
$n/4 + c\sqrt{n \log M}$) repeated bases. We show 
that a random set $\calS$ contains many such strands (approximately $\sqrt{M})$ and then prove that for random strands $S$ from $\calS''$, the expected cost $\cost_R(S)$ is at least $2.5 n + c\sqrt{n \log M}$. This gives us a lower bound of $2.5 nk + c\sqrt{n \log M} - O(k\sqrt{n \log k})$ on the total cost of synthesising $k$ batches (note, typically $M \gg k$).

%% file: single_batch.tex
%%%%%%%%%%%%%%%%%%%%%%%%%%%%%%%%%%%%%%%%%%%%
\section{Single batch analysis} \label{sec:single_batch}
%%%
%%%%%%%%%%%%%%%%%%%%%%%%%%%%%%%%%%%%%%%%%%%%

As a warm-up to the multiple batch setting, in this section we analyze the single batch setting, that is, the setting where $k = 1$. 
As discussed in Section~\ref{sec:results}, by using the periodic strand $(ACGT)^{*}$ as a reference strand, 
we obtain (assuming $M \geq n$) the upper bounds 
$\cost \left( \cS \right) \leq 2.5 n + 3 \sqrt{n \log M}$ 
for unconstrained strands and 
$\cost \left( \cS \right) \leq 2 n + 3 \sqrt{n \log M}$ 
for strands without homopolymers, 
with both bounds holding with probability $1-o(1)$. 
The formal statement and its short proof are in Section~\ref{sec:single_batch_UB}.

We also obtain matching lower bounds, for both choices of the strand universe $\cU$, 
by an appropriate stochastic domination argument that compares random walks. 
The formal statements are in Section~\ref{sec:single_batch_LB} 
and their proofs are in Section~\ref{sec:single_batch_LB_proofs}.

%%%%%%%%%%%%%%%%%%%%%%%%%%%%%%%%%%%%%%%%%%%
\subsection{Upper bounds for a single batch}\label{sec:single_batch_UB} %%%
%%%%%%%%%%%%%%%%%%%%%%%%%%%%%%%%%%%%%%%%%%%

\begin{theorem}\label{thm:single_batch_UB}
Consider the problem setup in Section~\ref{sec:problem_statement} with $k=1$. Let $M' := \max \left\{ M, n \right\}$. 
\begin{enumerate}[(a)]
\item \emph{\textbf{(Unconstrained strands)}} 
With probability at least $1-1/n$ we have that 
\[
\cost \left( \cS \right) \leq 2.5 n + 3 \sqrt{n \log M'}.
\]

\item \emph{\textbf{(Strands without homopolymers)}} 
With probability at least $1-1/n$ we have that 
\[
\cost \left( \cS \right) \leq 2 n + 3 \sqrt{n \log M'}.
\]
\end{enumerate}
\end{theorem}
\begin{proof} 
Consider first the case of unconstrained strands. We fix $R:= (ACGT)^{*}$ as the reference strand with which we print the strands in $\cS$. 
For a strand $S \in \cU$, let $\tau_{i} \left( S \right)$ denote the index of the base of $R$ that is used to print the $i$th base of $S$. With this notation, $R$ needs $\max_{S \in \cS} \tau_{n} \left( S \right)$ bases to print all strands in $\cS$. 
This shows that $\cost \left( \cS \right) \leq \cost_{R} \left( \cS \right) = \max_{S \in \cS} \tau_{n} \left( S \right)$. 
Let $\lambda := 3 \sqrt{n \log M'}$ and $m := 2.5 n + \lambda$. 
Combining the previous observation with a union bound we thus have that 
\[
\p \left\{ \cost \left( \cS \right) > m \right\}
\leq 
\p \left\{ \cost_{R} \left( \cS \right) > m \right\}
= 
\p \left\{ \max_{S \in \cS} \tau_{n} \left( S \right) > m \right\}
\leq M \p \left\{ \tau_{n} \left( S \right) > m \right\}.
\]
The key observation is that if $S$ is a uniformly random strand, 
then 
$\tau_{n} \left( S \right) \overset{d}{=} \sum_{i=1}^{n} X_{i}$, 
where $\left\{ X_{i} \right\}_{i=1}^{n}$ are i.i.d.\ random variables that are uniform on $\left\{ 1, 2, 3, 4 \right\}$; here $\overset{d}{=}$ denotes equality in distribution. 
Therefore, noting that $\E \left[ X_{1} \right] = 2.5$, by Hoeffding's inequality (Theorem~\ref{thm:hoeffding}) we have that 
\[
\p \left\{ \tau_{n} \left( S \right) > m \right\} 
= \p \left\{ \sum_{i=1}^{n} \left( X_{i} - \E \left[ X_{i} \right] \right) > \lambda \right\}
\leq \exp \left( - \frac{2\lambda^{2}}{9n} \right) = \frac{1}{\left( M' \right)^{2}}.
\]
Combining the two displays above we have that 
$\p \left\{ \cost \left( \cS \right) > m \right\} 
\leq M / \left( M' \right)^{2} 
\leq 1/n$, as desired. 

The case of strands without homopolymers is analogous, 
the only change is that now 
$\left\{ X_{i} \right\}_{i=1}^{n}$ are independent random variables with $X_{1}$ uniformly random on $\left\{ 1, 2, 3, 4 \right\}$ and $X_{i}$ uniformly random on $\left\{ 1, 2, 3 \right\}$ for $i \geq 2$. 
\end{proof}

%%%%%%%%%%%%%%%%%%%%%%%%%%%%%%%%%%%%%%%%%%%
\subsection{Lower bounds for a single batch}\label{sec:single_batch_LB}\label{sec:lb_single_batch} %%%
%%%%%%%%%%%%%%%%%%%%%%%%%%%%%%%%%%%%%%%%%%%

The lower bounds for the two strand universes are analogous, but we state them separately for clarity. 
In both cases we present two bounds: one focusing on the constant of the linear term, with weak assumptions on $M$, 
the other focusing on the second order term, assuming slightly more about $M$. 

%%%%%%%%%%%%%%%%%%%%%%%%%%%%%%%%%%%%%%%%%%%%
\subsubsection{Unconstrained strands} \label{sec:single_batch_LB_general} %%%
%%%%%%%%%%%%%%%%%%%%%%%%%%%%%%%%%%%%%%%%%%%%

Recall the upper bound of $2.5n + 3 \sqrt{n \log M}$ from Theorem~\ref{thm:single_batch_UB} (assuming $M \geq n$).
Our first lower bound result says that the constant $2.5$ cannot be improved, even if the batch only consists of a (large enough) constant amount of strands.

\begin{theorem}\label{thm:lb_single_batch_general_basic} 
Consider the problem setup in Section~\ref{sec:problem_statement} with $k=1$ and unconstrained strands. 
Fix $\eps > 0$ and let $M \geq 21/\eps^{2}$.
Then, with probability at least $1 - \exp\left( - n \right)$ 
we have that 
\[
\cost \left( \cS \right) \geq \left( 2.5 - \eps \right) n.
\]
\end{theorem}
This result can be significantly sharpened: if the number of strands in $\cS$ is at least linear in $n$ (and at most exponential in $n$), 
then not only is the main term $2.5 n$ required in the cost, 
but even the additional $\sqrt{n \log M}$ term is necessary. 

\begin{theorem}\label{thm:lb_single_batch_general_sqrtn} 
Consider the problem setup in Section~\ref{sec:problem_statement} with $k=1$ and unconstrained strands. 
Let $M$ satisfy $\left( 5 \exp \left( 45 \right) \right) n \leq M \leq 5n \exp \left( 4 n / 25 \right)$. 
Then, with probability at least $1-\exp \left( - n \right)$ 
we have that 
\[
\cost \left( \cS \right) \geq 
2.5 n + \frac{1}{5} \sqrt{n \log \left( \frac{M}{5n} \right)}.
\]
\end{theorem}

\begin{remark}[Expected Shortest Common Supersequence ($\mathsf{SCS}$)] \label{rem:SCS}
Since $\cost(\calS) = \mathsf{SCS}(\calS)$, our results also provide bounds on the SCS of a set of strings. Jiang and Li consider the case when $M=n$, and they prove that $\E_\calS[\mathsf{SCS}(\calS)] = 2.5n \pm O(n^{0.707})$ for a set $\calS$ of $n$ uniformly random length $n$ quaternary strings
\cite[Corollary 4.11]{JiangLi1995}. Combining Theorem~\ref{thm:single_batch_UB} and Theorem~\ref{thm:lb_single_batch_general_basic} with $\epsilon = \sqrt{21/n}$ we obtain  
$2.5n - \sqrt{21n} \leq \E_\calS[\mathsf{SCS}(\calS)] \leq 2.5n + 3\sqrt{n \log n},$ tightening the prior result. 
\end{remark}

%%%%%%%%%%%%%%%%%%%%%%%%%%%%%%%%%%%%%%%%%%%%
\subsubsection{Strands without homopolymers} \label{sec:lb_single_batch_no_homopolymers} %%%
%%%%%%%%%%%%%%%%%%%%%%%%%%%%%%%%%%%%%%%%%%%%

Next, we state theorems analogous to those in Section~\ref{sec:single_batch_LB_general}, but in the constrained setting where $\cU$ contains strands without homopolymers. (We allow the reference strand to potentially have homopolymers.)
Recall the upper bound of $2n + 3 \sqrt{n \log M}$ from Theorem~\ref{thm:single_batch_UB} (assuming $M \geq n$).
Our first lower bound result says that the constant $2$ cannot be improved, even if the batch only consists of a (large enough) constant amount of strands.

\begin{theorem}\label{thm:lb_single_batch_no_homopolymers_basic}
Consider the problem setup in Section~\ref{sec:problem_statement} with $k=1$ and strands without homopolymers.  
Fix $\eps > 0$ and let $M \geq 9 / \eps^{2}$. 
Then, with probability at least $1 - \exp\left( - n \right)$ 
we have that 
\[
\cost \left( \cS \right) \geq \left( 2 - \eps \right) n.
\] 
\end{theorem}

This result can be significantly sharpened: 
if the number of strands in $\cS$ is at least linear in $n$ (and at most exponential in $n$), then not only is the main term $2n$ required in the cost, 
but even the additional $\sqrt{n \log(M)}$ term is necessary.

\begin{theorem}\label{thm:lb_single_batch_no_homopolymers_sqrtn}
Consider the problem setup in Section~\ref{sec:problem_statement} with $k=1$ and strands without homopolymers.  
Let $M$ satisfy $\left( 5 \exp \left( 45 \right) \right) n \leq M \leq 5n \exp \left( 4 n / 25 \right)$. 
Then, with probability at least $1-\exp \left( - n \right)$ 
we have that 
\[
\cost \left( \cS \right) \geq 
2 n + \frac{3}{20} \sqrt{n \log \left( \frac{M}{5n} \right)}.
\] 
\end{theorem}

%%%%%%%%%%%%%%%%%%%%%%%%%%%%%%%%%%%%%%%%%%%%
\subsection{Lower bound proofs for a single batch} \label{sec:single_batch_LB_proofs} %%%
%%%%%%%%%%%%%%%%%%%%%%%%%%%%%%%%%%%%%%%%%%%%

We start with some definitions. 
Let $\Sigma := \left\{ A, C, G, T \right\}$. 
For two strands $R, S \in \Sigma^{*}$, let $\cE_{R} \left( S \right)$ denote the event that $R$ is a superstring of $S$.
Let $\cE_{R} \left( \cS \right)$ denote the event that $R$ is a superstring of all strands in $\cS$, that is,
$\cE_{R} \left( \cS \right) := \cap_{S \in \cS} \cE_{R} \left( S \right)$.
For an integer $m$, let $\cE_{m} \left( \cS \right)$ denote the event that there exists a strand $R \in \Sigma^{\leq m}$ such that the event $\cE_{R} \left( \cS \right)$ holds.
That is,
\[
\left\{ \cost \left( S \right) \leq m \right\} 
= 
\cE_{m} \left( \cS \right)
= \bigcup_{R \in \Sigma^{\leq m}} \cE_{R} \left( \cS \right)
= \bigcup_{R \in \Sigma^{m}} \cE_{R} \left( \cS \right),
\]
where the second equality holds because if $R$ is a superstring of $S$, then all superstrings of $R$ are also a superstring of $S$.

\subsubsection{Proofs for unconstrained strands}

\begin{proof}[Proof of Theorem~\ref{thm:lb_single_batch_general_basic}]
Fix $m := \left( 2.5 - \eps \right) n$.
We start with a union bound:
\begin{equation}\label{eq:union_bound}
\p \left\{ \cost \left( \cS \right) \leq m \right\} 
= 
\p \left\{\cE_{m} \left( \cS \right) \right\}
\leq \sum_{R \in \Sigma^{m}} \p \left\{ \cE_{R} \left( \cS \right) \right\}
= \sum_{R \in \Sigma^{m}} \left( \p \left\{ \cE_{R} \left( S \right) \right\} \right)^{M},
\end{equation}
where the equality is due to the fact that the strands in $\cS$ are i.i.d.
Our goal now is to bound $\p \left\{ \cE_{R} \left( S \right) \right\}$,
where $R \in \Sigma^{m}$ is fixed and $S \in \Sigma^{n}$ is uniformly random.

To understand the probability of this event we introduce some notation.
First, we extend the reference strand $R$ indefinitely, by concatenating the strand $\left( ACGT \right)^{*}$ to the end of $R$---this is done just so that everything in the following is well-defined---we refer to this extended strand as~$R'$.
Note that $R'$ is a superstring of $S$,
and the original reference strand $R$ is a superstring of $S$ if and only if
the printing of $S$ using $R'$ succeeds in at most $m$ steps.
% (i.e., it uses at most $m$ bases of $R'$).
Let $\tau_{i} \left( S, R' \right)$ denote the index of the base of $R'$ that is used to print the $i$th base of $S$.
We also define $X_{1} \left( S, R' \right) := \tau_{1} \left( S, R' \right)$,
and $X_{i} \left( S, R' \right) := \tau_{i} \left( S, R' \right) - \tau_{i-1} \left( S, R' \right)$ for $i > 1$.
With this notation we have that
\[
\cE_{R} \left( S \right)
= \left\{ \tau_{n} \left( S, R' \right) \leq m \right\}
= \left\{ \sum_{i=1}^{n} X_{i} \left( S, R' \right) \leq m \right\}.
\]
Let $\left\{ Y_{i} \right\}_{i = 1}^{n}$ be i.i.d.\ random variables that are uniform on $\left\{ 1, 2, 3, 4 \right\}$.
Note that for an arbitrary $R$ the random variables $\left\{ X_{i} \left( S, R' \right) \right\}_{i=1}^{n}$ are not i.i.d.---in fact, they are not even necessarily independent.
Specifically,
given $R'$ and
$\left\{ X_{j} \left( S, R' \right) \right\}_{j=1}^{i-1}$,
the \emph{support} of $X_{i} \left( S, R'  \right)$
is determined.
However, no matter what, this support always consists of four distinct positive integers---the distances to the next occurrences of the four bases $A$, $C$, $G$, and $T$ in $R$.
Moreover, the \emph{distribution} on this support is always uniform---this is because,
given $R'$ and
$\left\{ X_{j} \left( S, R' \right) \right\}_{j=1}^{i-1}$,
the value of $X_{i} \left( S, R' \right)$ is determined by $S_{i}$, which is uniformly random on $\Sigma$.
Since at best (in terms of minimization) the four distinct positive integers in the support of $X_{i} \left( S, R' \right)$ are $1$, $2$, $3$, and $4$, the random variable $X_{i} \left( S, R' \right)$ \emph{stochastically dominates} $Y_{i}$.
Moreover, since the bases $\left\{ S_{i} \right\}_{i = 1}^{n}$ are independent,
we also have that
$\sum_{i=1}^{n} X_{i} \left( S, R' \right)$
stochastically dominates
$\sum_{i=1}^{n} Y_{i}$,
and so we have that
\begin{equation}\label{eq:stochastic_domination}
\p \left\{ \cE_{R} \left( S \right) \right\}
= \p \left\{ \sum_{i=1}^{n} X_{i} \left( S, R' \right) \leq m \right\}
\leq \p \left\{ \sum_{i=1}^{n} Y_{i} \leq m \right\}.
\end{equation}
We can now bound this latter quantity using standard estimates. In particular, using Hoeffding's inequality (Theorem~\ref{thm:hoeffding}) we obtain that
\[
\p \left\{ \sum_{i=1}^{n} Y_{i} \leq m \right\}
\leq \exp \left\{ - \frac{2}{9} \eps^{2} n \right\}.
\]
Plugging this back into~\eqref{eq:union_bound} and~\eqref{eq:stochastic_domination}, we have obtained the bound
\begin{equation}\label{eq:single_batch_bound}
\p \left\{ \cost \left( \cS \right) \leq m \right\}
\leq 4^{m} \exp \left( - \frac{2}{9} \eps^{2} M n \right)
\leq \exp \left( \left\{ 2.5 \log 4 - \frac{2}{9} \eps^{2} M \right\} n \right).
\end{equation}
If $M \geq 21/\eps^{2}$, then this is at most $\exp \left( - n \right)$.
\end{proof}

\begin{proof}[Proof of Theorem~\ref{thm:lb_single_batch_general_sqrtn}]
Let 
$m := 2.5n + \lambda$,  
where 
$\lambda = \frac{1}{5} \sqrt{n \log \left( \frac{M}{5n} \right)}$. 
The proof is identical to the proof of Theorem~\ref{thm:lb_single_batch_general_basic} until~\eqref{eq:stochastic_domination}. 
At this point in the proof, we bound this probability differently. 
Specifically, by applying Lemma~\ref{lem:right_tail_LB} with $\ell = 4$, we have that 
\begin{equation}\label{eq:LB_on_right_tail_single_batch}
\p \left\{ \sum_{i=1}^{n} Y_{i} > m \right\} 
\geq \exp \left( - 25 \lambda^{2} / n \right)
= \frac{5n}{M}.
\end{equation}
In order to obtain this bound via Lemma~\ref{lem:right_tail_LB} we must have that 
$(4/3) \sqrt{n} \leq \lambda \leq 2n/25$;  
these inequalities hold due to the condition 
$\left( 5 \exp \left( 45 \right) \right) n \leq M \leq 5n \exp \left( 4 n / 25 \right)$ assumed in Theorem~\ref{thm:lb_single_batch_general_sqrtn}. 
Note in particular that we thus have $m \leq 2.58 n$. 
From~\eqref{eq:LB_on_right_tail_single_batch} we thus have that 
\[
\p \left\{ \sum_{i=1}^{n} Y_{i} \leq m \right\} 
\leq 
1 - \frac{5n}{M}.
\]
Plugging this into~\eqref{eq:union_bound} and~\eqref{eq:stochastic_domination}, 
we have obtained---analogously to~\eqref{eq:single_batch_bound}---the bound 
\[
\p \left\{ \cost \left( \cS \right) \leq m \right\}
\leq 4^{m} \left( 1 - \frac{5n}{M} \right)^{M} 
\leq \exp \left( \left( 2.58 \log 4 \right) n - 5n \right) 
\leq \exp \left( - n \right). \qedhere
\]
\end{proof}

\subsubsection{Proofs for strands without homopolymers}

\begin{proof}[Proof of Theorem~\ref{thm:lb_single_batch_no_homopolymers_basic}]
The proof is almost identical to the proof of Theorem~\ref{thm:lb_single_batch_general_basic}, so we only highlight the changes.
First, we set $m := \left( 2 - \eps \right) n$.
Second,
$\left\{ Y_{i} \right\}_{i=1}^{n}$ are now i.i.d.\ random variables that are uniform on $\left\{ 1, 2, 3 \right\}$. Thus Hoeffding's bound gives that
\[
\p \left\{ \sum_{i=1}^{n} Y_{i} \leq m \right\}
\leq \exp \left( - \frac{1}{2} \eps^{2} n \right).
\]
The rest of the proof is identical.
\end{proof}

\begin{proof}[Proof of Theorem~\ref{thm:lb_single_batch_no_homopolymers_sqrtn}]
Let 
$m := 2n + \lambda$,  
where 
$\lambda = \frac{3}{20} \sqrt{n \log \left( \frac{M}{5n} \right)}$. 
The proof is identical to the proof of 
Theorem~\ref{thm:lb_single_batch_general_sqrtn}, 
except  
$\left\{ Y_{i} \right\}_{i=1}^{n}$ are now i.i.d.\ random variables that are uniform on $\left\{ 1, 2, 3 \right\}$, 
and thus Lemma~\ref{lem:right_tail_LB} is applied with $\ell = 3$. 
\end{proof}

%% file: quantile-bounds.tex
\section{Empirical Quantiles}\label{sec:emp-quant}

Consider a distribution $D$ on strands, and 
a reference strand $R$. For $q\in(0,1)$, we define the $q$-quantile of the distribution of the printing cost $\cost_R(S)$, where $S$ is distributed according to $D$, as the minimum $t$ such that $\p\left\{\cost_R(S) \leq t \right\} \geq q$:  
$$Q_{q,R}(D) := \min \Big\{t :\p\big\{\cost_R(S) \leq t\big\}\geq q\Big\}.$$
Similarly, for a set of strands $\calS$, we define an \emph{empirical} variant of $Q$ as:
$$\widetilde Q_{q,R}(\calS) := \min \Big\{t:\frac{|\{S\in \calS: \cost_R(S) \leq t\big)\}|}{|\calS|}\geq q\Big\}.$$
For a family of reference strands $\calR$, we let 
$$Q_{q,\calR}(D) := \min_{R\in \calR} Q_{q,R}(D)
\qquad \text{and} \qquad  
\widetilde Q_{q,\calR}(\calS) := \min_{R\in \calR} \widetilde Q_{q,R}(\calS).$$

Consider a random set $\calS_{D,M}$ that contains 
$M$ i.i.d.\ samples from the distribution $D$. By the Glivenko–Cantelli theorem, we have
$\EQ_{q,R}(\calS_{D,M})\approx Q_{q,R}(D)$ for every $R$ and a sufficiently large~$M$. We use the Dvoretzky-Kiefer-Wolfowitz inequality to get a quantitative bound on 
$\EQ_{q,R}(\calS_{D,M})$. 
\begin{lemma} \label{lem:DKW}
For every distribution $D$ on strands, every reference strand $R$, and every positive $\varepsilon$, we have
\begin{equation}\label{eq:DKW:R}
\p\Big\{Q_{q-\varepsilon,R}(D)\leq \EQ_{q,R}(\calS_{D,M})\leq Q_{q+\varepsilon,R}(D)\text{ for all } q\in(\varepsilon,1-\varepsilon)\Big\} > 
1 - 2e^{-2M\varepsilon^2}.
\end{equation}
Furthermore, for every family of reference strands $\calR$, we have
\begin{equation}\label{eq:DKW:setR}
\p\Big\{Q_{q-\varepsilon,\calR}(D)\leq \EQ_{q,\calR}(\calS_{D,M})\leq Q_{q+\varepsilon,\calR}(D)\text{ for all } q\in(\varepsilon,1-\varepsilon)\Big\} > 
1 - 2|\calR|e^{-2M \varepsilon^2}.    
\end{equation}
\end{lemma}
\medskip

\noindent We prove this lemma in Appendix~\ref{sec:proof:lem:DKW}.

%% file: ub_many_batches.tex
%%%%%%%%%%%%%%%%%%%%%%%%%%%%%%%%%%%%%%%%%%%%
\section{Upper Bounds for Multiple Batches} \label{sec:ub_multi_batch} %%%
%%%%%%%%%%%%%%%%%%%%%%%%%%%%%%%%%%%%%%%%%%%%

\subsection{Batching for strands without homopolymers}\label{sec:ub_multi_batch_no_homopoly}

Our approach is to define a quantile-based batching process and to then split the strands into two groups based on whether we use the reference strand $\widetilde R = (ACGT)^*$ or its reverse $\widebar R = (TCGA)^*$. We first observe that one of these two options will lead to a cost of at most $2n$ if the strand does not contain homopolymers (and this result holds deterministically).

\begin{lemma}\label{lem:cost-2n}
Let $\widetilde R = (ACGT)^*$ and $\widebar R = (TCGA)^*$ be the alternating sequence and its reverse.
For a strand $S \in \ACGT^n$ without homopolymers, we have that
\begin{equation}\label{eq:symmetry_identity} 
\cost_{\widetilde R}(S) + \cost_{\widebar R}(S) = 4n+1.
\end{equation}
%Moreover, there exists a bijection between strands $S$ with $\cost_{\widetilde R}(S) \leq 2n$ and $S'$ with  $\cost_{\widetilde R}(S') > 2n$.
Moreover, there exists a bijection $\varphi : \ACGT^n \to \ACGT^n$ such that for every strand $S \in \ACGT^n$ we have that 
\begin{equation}\label{eq:bijection_identity}
\cost_{\widetilde R}(S) + \cost_{\widetilde R}( \varphi(S) ) = 4n+1.
\end{equation}
\end{lemma}
\begin{proof}
Let $\tau_i(S, \widetilde R)$ and $\tau_i(S, \widebar R)$ be the time that the $\nth{i}$ character of $S$ is printed using $\widetilde R$ or $\widebar R$, respectively. 
Consider the per-character costs
$X_i(S,\widetilde R) = \tau_i(S,\widetilde R) - \tau_{i-1}(S,\widetilde R)$ 
and 
$X_i(S,\widebar R) = \tau_i(S,\widebar R) - \tau_{i-1}(S,\widebar R).$
Observe that $X_i(S,\widetilde R) + X_i(S,\widebar R) = 4$ for $i > 1$ and $X_1(S,\widetilde R) + X_1(S,\widebar R) = 5$. Hence,
$$
\cost_{\widetilde R}(S) + \cost_{\widebar R}(S) = \sum_{i=1}^n (X_i(S,\widetilde R) + X_i(S,\widebar R))  = 4n + 1.$$

We now map every strand $S$ to its compliment by replacing each base $A$ with $T$, $C$ with $G$, $G$ with $C$, and $T$ with $A$. Observe that if we renamed each base as above both in $S$ and the reference strand $\widetilde R$, then the cost would not change. That is, $\cost_{\widetilde R}(S) = \cost_{\widebar R}(\varphi(S))$. Using~\eqref{eq:symmetry_identity} we thus obtain~\eqref{eq:bijection_identity}. 
%the first item of the lemma, we get $$\cost_{\widetilde R}(S) + \cost_{\widetilde R}(\varphi(S)) = 4n+1.$$
\end{proof}

The key ideas in the following proof are to (i) use a quantile-based batching process to group the strands, and then (ii) decide whether they will be printed with $\widetilde R$ or $\widebar R$ to reduce the overall cost. By using the above lemma, we know that for strands that have cost larger than $2n$ under~$\widetilde R$, they will have cost at most $2n$ under $\widebar R$. Using this and properties of the batching process construction, we argue that for $k-2$ of the batches the cost is at most $2n$, with a constant fraction of these batches having cost at most $2n-\Omega(\sqrt{n})$. For the two ``middle'' batches, we only show that the sum of their costs is at most $4n+1$. Overall, we achieve a total cost of $2nk-\Omega(k\sqrt{n})$ when summing over the $k$ batches.

\begin{proof}[Proof of Theorem~\ref{thm:ub-multi}(1).]
For the proof, we assume that $k$ divides $M$; otherwise, we could set the batches to have size within $M/k \pm 1$. Let $\mathcal{D}$ denote the uniform distribution of length $n$ strands without homopolymers.

We start by defining the batching process and the assignment of reference strands (i.e., choosing between $\widetilde R$ and $\widebar R$ for each batch). To group the strands, we first sort the strands in $\calS$ in a nondecreasing order according to their cost with respect to $\widetilde R$. Then, we let $\calB_i$ be the subset of strands placed between positions $(i-1)M/k + 1$ and $iM/k$ (inclusive) in the ordering (e.g., $\calB_1$ contains the $M/k$ lowest cost strands). After the batches have been defined, we use $\widetilde R$ to print the first $\ell := \lceil k/2 \rceil$ batches and we use $\widebar R$ for the remainder. 

By Lemma~\ref{lem:cost-2n}, the batches with higher cost under $\widetilde R$ have lower cost under $\widebar R$. 
Furthermore, with high probability, all batches have cost at most $2n$ except perhaps the ``middle'' two batches~$\calB_{\ell}$ and~$\calB_{\ell+1}$.
More precisely, we utilize the empirical quantiles $i/k$ for $i\in[k]$, where we recall that~$\EQ_{i/k, \widetilde R}(\calS)$ denotes the minimum value $t_i$ such that an $i/k$ fraction of strands in $\calS$ have cost at most~$t_i$ with respect to $\widetilde R$. In particular, the bound 
\begin{equation}\label{eq:cost_bd_left}
\cost_{\widetilde R}(\calB_i) \leq \EQ_{i/k, \widetilde R}(\calS)
\end{equation}
for $i \in \{ 1, \ldots, \ell -1 \}$ follows from the batch construction process. Similarly, we have that 
\begin{equation}\label{eq:cost_bd_right}
\cost_{\widebar R}(\calB_{k-i}) \leq 4n + 1 - \EQ_{(k-i-1)/k, \widetilde R}(\calS)
\end{equation}
for $i \in \{ 0, 1, \ldots, k - \ell - 2 \}$, due to Lemma~\ref{lem:cost-2n} and the batch construction process. 

The bounds in~\eqref{eq:cost_bd_left} and~\eqref{eq:cost_bd_right} both involve empirical quantiles, for which Lemma~\ref{lem:DKW} provides uniform bounds. 
We apply Lemma~\ref{lem:DKW} with $\varepsilon = 1/(4k)$. 
This implies that, with probability at least $1-2\exp \left( - M / (8k^{2}) \right)$, we have that 
\[
\EQ_{i/k, \widetilde R}(\calS) 
\leq Q_{(i+\nicefrac{1}{4})/k, \widetilde R}(\mathcal{D})
\] 
for all $i \in \{ 1, \ldots, \ell - 1 \}$, and furthermore that 
\[
\EQ_{(k-i-1)/k, \widetilde R}(\calS) 
\geq Q_{(k-i-\nicefrac{5}{4})/k, \widetilde{R}}(\mathcal{D})
\]
for every $i \in \{ 0, 1, \ldots, k-\ell-2\}$. 
Plugging these bounds into~\eqref{eq:cost_bd_left} and~\eqref{eq:cost_bd_right} 
we obtain that, 
with probability at least $1-2\exp \left( - M / (8k^{2}) \right)$, 
the total cost is bounded above by 
\begin{multline}\label{eq:cost_bound}
\sum_{i=1}^{\ell} \cost_{\widetilde R}(\calB_i) 
+ \sum_{i=\ell+1}^{k} \cost_{\widebar R}(\calB_i) \\
\leq \sum_{i=1}^{\ell - 1} Q_{(i+\nicefrac{1}{4})/k, \widetilde R}(\mathcal{D}) 
+ \left\{ \cost_{\widetilde R}(\calB_\ell) + \cost_{\widebar R}(\calB_{\ell+1}) \right\} 
+ \sum_{i=0}^{k-\ell-2} \left( 4n + 1 - Q_{(k-i-\nicefrac{5}{4})/k, \widetilde{R}}(\mathcal{D}) \right).
\end{multline}
In the rest of the proof we bound from above these two sums, as well as the term in the middle. 

First, we claim that for every fixed $\delta > 0$ there exists $\alpha = \alpha(\delta) > 0$ such that 
\begin{equation}\label{eq:quantile_bound}
Q_{\nicefrac{1}{2} - \delta, \widetilde{R}}(\mathcal{D}) 
\leq 2n - \alpha \sqrt{n}. 
\end{equation}
To see this, recall that the cost of printing a strand $S$ using $\widetilde{R}$ is a sum of independent random variables: the cost of the first character is uniform in $\{1,2,3,4\}$ and the cost of the remaining characters are uniform in $\{1,2,3\}$ (since the strand $S$ does not have homopolymers). Therefore $\cost_{\widetilde{R}}(S)$ is approximately Gaussian with mean $2n+1/2$ and variance on the order of $n$. This means that if we consider a quantile that is bounded away from the median, then the cost is smaller than the mean by at least a constant factor of the standard deviation. Since the standard deviation is on the order of $\sqrt{n}$, this implies~\eqref{eq:quantile_bound}. 

Turning back to~\eqref{eq:cost_bound}, consider the first $\lceil k/3 \rceil$ terms of the first sum in~\eqref{eq:cost_bound} (note that $\lceil k/3 \rceil \geq 1$ since $k \geq 3$).  
Note that 
$\left( \lceil k/3 \rceil + 1/4 \right) / k \leq 5/12 = 1/2 - 1/12$ for all $k \geq 3$. 
Let $\alpha_{*} := \alpha \left( 1/12 \right) > 0$. 
Then~\eqref{eq:quantile_bound} implies that 
for every $i \in \{ 1, \ldots, \lceil k/3 \rceil \}$ we have that 
\[
Q_{(i+\nicefrac{1}{4})/k, \widetilde R}(\mathcal{D}) 
\leq 
Q_{(\lceil k/3 \rceil+\nicefrac{1}{4})/k, \widetilde R}(\mathcal{D}) 
\leq 
Q_{5/12, \widetilde R}(\mathcal{D}) 
\leq 2n - \alpha_{*} \sqrt{n}.
\]
Note also that for every $i \in \{ 1, \ldots, \ell - 1 \}$ we have that 
$(i+1/4)/k \leq (\ell - 3/4)/k \leq (k/2-1/4)/k = 1/2 - 1/(4k) < 1/2$. 
Together with~\eqref{eq:quantile_bound}, this implies that for every $i \in \{ 1, \ldots, \ell - 1 \}$ we have that 
\[
Q_{(i+\nicefrac{1}{4})/k, \widetilde R}(\mathcal{D}) \leq 2n.
\]
Putting the bounds in the previous two displays together, we obtain that 
\begin{align}
\sum_{i=1}^{\ell - 1} Q_{(i+\nicefrac{1}{4})/k, \widetilde R}(\mathcal{D}) 
&= \sum_{i=1}^{\lceil k/3 \rceil} Q_{(i+\nicefrac{1}{4})/k, \widetilde R}(\mathcal{D}) 
+ \sum_{\lceil k/3 \rceil + 1}^{\ell - 1} Q_{(i+\nicefrac{1}{4})/k, \widetilde R}(\mathcal{D}) \notag \\
&\leq \sum_{i=1}^{\lceil k/3 \rceil} \left( 2n - \alpha_{*} \sqrt{n} \right) + \sum_{\lceil k/3 \rceil + 1}^{\ell - 1} 2n 
\leq 2 (\ell -1) n - (\alpha_{*}/3) k \sqrt{n}. \label{eq:gain}
\end{align} 
Similarly, note that for every $i \in \{0, 1, \ldots, k-\ell-2 \}$ we have that 
$(k-i-5/4)/k \geq 1/2 + 1/(4k) > 1/2$. 
Therefore, due to the symmetry of the distribution of $\cost_{\widetilde{R}}(S)$, 
we have that 
$Q_{(k-i-\nicefrac{5}{4})/k, \widetilde{R}}(\mathcal{D}) \geq 2n+1$ 
for every $i \in \{0, 1, \ldots, k-\ell-2 \}$.  Using this bound in the second sum of~\eqref{eq:cost_bound}, together with~\eqref{eq:gain}, we obtain that the quantity in~\eqref{eq:cost_bound} is bounded from above by 
\begin{equation}\label{eq:cost_bd_almost_done}
2(k-2)n - (\alpha_{*}/3) k \sqrt{n} 
+ \left\{ \cost_{\widetilde R}(\calB_\ell) + \cost_{\widebar R}(\calB_{\ell+1}) \right\}. 
\end{equation}
Finally, we bound the sum in the curly brackets in~\eqref{eq:cost_bd_almost_done}. 
Let $S'$ be a strand in $\mathcal{B}_{\ell}$ such that 
$\cost_{\widetilde{R}}(\mathcal{B}_{\ell}) = \cost_{\widetilde{R}}(S')$, 
and let $S''$ be a strand in $\mathcal{B}_{\ell+1}$ such that 
$\cost_{\widebar{R}}(\mathcal{B}_{\ell+1}) = \cost_{\widebar{R}}(S'')$. 
By the construction of the batching process we have that 
$\cost_{\widetilde{R}}(S') \leq \cost_{\widetilde{R}}(S'')$. 
On the other hand, 
by Lemma~\ref{lem:cost-2n} we have that 
$\cost_{\widetilde{R}}(S'') + \cost_{\widebar{R}}(S'') = 4n+1$. 
Putting all this together we have that 
\[
\cost_{\widetilde R}(\calB_\ell) + \cost_{\widebar R}(\calB_{\ell+1}) 
= \cost_{\widetilde R}(S') + \cost_{\widebar R}(S'') 
\leq \cost_{\widetilde R}(S'') + \cost_{\widebar R}(S'') 
= 4n+1.
\]
Plugging this back into~\eqref{eq:cost_bd_almost_done}, we obtain that the total cost is bounded from above by 
$2kn + 1 - (\alpha_{*}/3)k\sqrt{n}$, as desired, with this bound holding with probability at least $1-2\exp \left( - M / (8k^{2}) \right)$. 
\end{proof}

\subsection{Batching for unrestricted strands}

The main idea of the proof is to analyze the cost quantiles after splitting into batches. In this case, the optimal splitting is a bit easier to define, as we will always use
the reference strand $\widetilde R = (ACGT)^{*}$. Once the reference strand is fixed, it is easy to see that for arbitrary strands, the optimal batching process involves first taking $M/k$ lowest cost strands and then the next $M/k$ lowest cost and so on.  

\begin{proof}[Proof of Theorem~\ref{thm:ub-multi}(2).]
In this argument, we consider a set $\calS$ of $M$ strands sampled i.i.d.\ from the uniform distribution $\mathcal{U}$ over $\ACGT^n$. For simplicity of notation, we let $\EQ_q = \EQ_{q,\widetilde R}(\calS)$ and $Q_q = Q_{q,\widetilde R}(\mathcal{U})$ denote the empirical and distributional $q$-quantiles of the cost, respectively. We utilize the empirical quantiles $i/k$ for $i\in[k]$, where we recall that $\EQ_{i/k}$ denotes the minimum value $t_i$ such that an $i/k$ fraction of strands in $\calS$ have cost at most $t_i$. In particular, the bound $\cost_{\widetilde R}(\calB_i) \leq \EQ_{i/k}$ follows from the batch construction process. 

Lemma~\ref{lem:DKW} provides a uniform bound on each empirical quantile, and we apply it with $\varepsilon = 1/k$; 
specifically, this implies that, 
with probability at least $1 - 2\exp(-\frac{2M}{k^2})$, 
we have for all $i \in [k-1]$ that 
$\EQ_{i/k} \leq Q_{(i+1)/k}$. 
In the following we assume that we are on this event; 
in particular, on this event we have that $\cost_{\widetilde R}(\calB_i) \leq Q_{(i+1)/k}$ holds for all $i \in [k-1]$. 

We next claim that for any even $k' \leq k-3$ we have that 
\begin{equation}\label{eq:no-hp-ub-sum}
%\sum_{i=1}^{k'}\cost_{\widetilde R}(\calB_i) \leq 
\sum_{i=1}^{k'} Q_{(i+1)/k} \leq 2.5nk'.
\end{equation}
Before showing~\eqref{eq:no-hp-ub-sum}, we conclude the proof assuming that it holds. 
When $k$ is even, we set $k' = k-4$, and otherwise, $k' = k-3$. 
Then, for $i$ in the range $k' < i \leq k$, we simply use the bound from Theorem~\ref{thm:single_batch_UB} that shows that 
$\cost_{\widetilde R}(\calB_i) \leq 2.5n + 3 \sqrt{n \ln M}$. 
Combining this with~(\ref{eq:no-hp-ub-sum}) implies the desired bound with $C_2 = 12$.

We now turn to proving~\eqref{eq:no-hp-ub-sum}. 
We have already seen in Lemma~\ref{lem:cost-2n} that the cost distribution is symmetric around $2.5n$ for strands uniform over $\ACGT^n$. This implies that 
\begin{equation}\label{eq:no-hp-ub-pair}
Q_{(i+1)/k} + Q_{(k-i-2)/k} \leq 2 \cdot 2.5n
\end{equation}
for each $i \leq k'/2$. Here we have paired up $(i+1)/k$ with $(k-i-2)/k$, and hence, we have chosen a cost quantile on either side of the mean, but with a slight asymmetry, shifting the larger one over by one.
The key observation is that $Q_{(k-i-2)/k}$ deviates from the mean less than $Q_{(i+1)/k}$.
More precisely, we have that $Q_{(i+1)/k} \leq 2.5n \leq Q_{(k-i-2)/k}$, and furthermore, 
$
Q_{(k-i-2)/k} - 2.5n \leq 2.5n - Q_{(i+1)/k}.
$
As this holds for each $i \leq k'/2$, the inequality in~(\ref{eq:no-hp-ub-sum}) follows.
\end{proof}

%% file: lb-4-batches.tex
\section{Lower Bounds}\label{sec:lb_multi_batch}

In this section we prove Theorem~\ref{thm:lb-4-homoplymers-main}, which provides lower bounds on the cost of batched DNA synthesis for random strands, for both unconstrained strands and strands without homopolymers.

We first introduce some notation and terminology. 
For a strand $S$, we say that a base $S_i$ is a
repetition of the previous base if this base is the same as the previous base, that is, if $S_i=S_{i-1}$. Let $d(S)$ be the number of bases that are repetitions of the previous
base in $S$ (that is, $d(S) := |\{i: S_i = S_{i+1}\}|$). For every $p\in [0,1]$ and positive integer $n$, define a distribution $D_{p, n}$ on DNA strands of length $n$ by letting the probability of a strand $S$ of length $n$ according to 
$D_{p,n}$ be equal to 
$$D_{p, n}(S) := \frac{1}{4}\cdot \Big(\frac{1-p}{3}\Big)^{n-d(S) - 1}p^{d(S)}.$$
We can think of this distribution as follows. The first base of the strand is chosen uniformly at random from the alphabet $\{A,C,G, T\}$. Then, every 
consecutive base is a repetition of the previous base with probability $p$ and another base with probability $(1-p)$. If a base is not a repetition of the previous base, then it is chosen uniformly at random among the remaining three bases.

In this paper, the two most important distributions on DNA strands are $D_{0,n}$ and $D_{\nicefrac{1}{4},n}$. The former is the uniform distribution on strands of length $n$ without homopolymers; the latter is the uniform distribution on strands of length $n$ (unconstrained, i.e., allowing homopolymers). In this section, we will also crucially use another distribution: $D_{\nicefrac{1}{4} + \delta,n}$ with
$\delta\approx \sqrt{\frac{\log M}{n}}$. To simplify notation, we will omit the second parameter of the distribution $D$ and write $D_p \equiv D_{p,n}$, since the length of all strands we consider is $n$.

We will prove the following theorem, from which Theorem~\ref{thm:lb-4-homoplymers-main} readily follows. 

\begin{theorem}\label{thm:lb-4-homoplymers}
Let $\calS \equiv \calS_{D_{p},M}$ be a set of $M$ i.i.d.\ strands from the distribution $D_{p}$, where $p$ is $0$ or $\nicefrac{1}{4}$.
Suppose that $M$ is divisible by $k$, and $k\leq \sqrt{M/(24n)}$. Then,
with probability at least $1 - 2\exp(- M / (4k^{2}) ) \geq 1 - 2 \exp(-6n)$,
the optimal cost of printing the strands in $\calS$ using $k$ batches of equal size is at least
\begin{align}
\label{eq:lb-without-hpol}
k\cdot(2n -\sqrt{5n\log (2k)}) \qquad \;\;\;\;\;\;&\text{ for strands without homopolymers;}\\
\label{eq:lb-with-hpol}
k\cdot(2.5n -\sqrt{5n\log (2k)}) \qquad \;\;\;\;\;\;&\text{ for unconstrained strands.}
\end{align}
Furthermore, there exist positive absolute constants $c_{1}$, $c_{2}$, and $c_{3}$ such that the following holds in the setting of unconstrained strands (that is, when $p =\nicefrac{1}{4}$). 
Suppose that the number of batches satisfies $k\leq c_1\min\{\sqrt{\log M/\log\log M},\sqrt{n/\log n},\sqrt[4]{M}/\sqrt{n}\}$. 
Then, with probability at least  $1-\exp(- c_{3} \sqrt{M}/k^2)$, 
the total cost of printing the strands in $\calS_{D_{\nicefrac{1}{4}},M}$ using $k$ batches of equal size is at least
\begin{equation}\label{eq:strong-lb-with-hpol}
k\cdot\Big(2.5n +c_2\frac{\sqrt{n\cdot \min\{\log M, n\}}}{k}\Big). 
\end{equation}
\end{theorem}
\begin{proof}
We use the approach outlined in the proof overview (see Section~\ref{sec:proof-overview}). We first prove the bounds~\eqref{eq:lb-without-hpol} and~\eqref{eq:lb-with-hpol}. To do so, we show that with high probability the cost of every batch that contains at least $n/k$ strands is greater than $2n -\sqrt{5n\log (2k)}$ for random strands without homopolymers and $2.5n -\sqrt{5n\log (2k)}$ for unconstrained random strands. Consequently, the total cost of printing $k$ batches is at least $k\cdot(2n - \sqrt{5n\log (2k)})$ and $k\cdot(2.5n - \sqrt{5n\log (2k)})$ for random strands without and with homopolymers, respectively. 

Let $\calR^*$ be the set of all \emph{reasonable} reference strands for printing strands of lengths $n$, that is, $\calR^*$ is the set of all possible strands of length $4n$ appended with the infinite repeating sequence $(ACGT)^*$ (see Section~\ref{sec:proof-overview}). The size of this set is $|\calR^{*}| = 4^{4n} < e^{6n}$. 
Recall  the notions of the $q$-quantile $Q_{q,\calR}(D)$ of the distribution $D$ and the 
empirical $q$-quantile $\EQ_{q,\calR}(\calS)$ of a sample $\calS$ that we introduced in Section~\ref{sec:emp-quant}. 
Since every batch in the optimal partitioning of the set $\calS$ contains a $\nicefrac{1}{k}$ fraction of all strands, its cost is at least 
$\EQ_{\nicefrac{1}{k},\calR^{*}}(\calS)$. 
In Lemma~\ref{lem:DKW} we showed that $\EQ_{q,\calR}(\calS)\geq Q_{q-\varepsilon,\calR}(D)$ with probability close to $1$ for sufficiently large $M$. Specifically, using Lemma~\ref{lem:DKW} with parameters 
$q = \nicefrac{1}{k}$, 
$\varepsilon=q/2$, and $\calR=\calR^*$, we obtain the following bound:
\begin{equation}\label{cor:lem:DKW}
\p\Big\{\EQ_{\nicefrac{1}{k},\calR^*}(\calS_{D,M})\geq Q_{\nicefrac{1}{(2k)},\calR^*}(D)\Big\} 
\geq 1 - 2|\calR^*| e^{-M/(2k^{2})} > 
1 - 2e^{-M/(4k^{2})},     
\end{equation} 
where in the second inequality we used that $|\calR^*| < \exp(6n) \leq e^{M/(4k^{2})}$; 
this holds due to the condition $k\leq \sqrt{M/(24n)}$ that is assumed in the statement of  
Theorem~\ref{thm:lb-4-homoplymers}.

\medskip

We now obtain lower bounds on $Q_{q,\calR^*}(D_{0})$, 
$Q_{q,\calR^*} (D_{\nicefrac{1}{4}})$ and $Q_{q,\calR^*} (D_{\nicefrac{1}{4}+\delta})$. 
\begin{lemma}\label{lem:lb-Q}
For $q \in (0,1)$ we have that 
\begin{align}
\label{eq:lb-q0}
    Q_{q,\calR^{*}}(D_{0})&\geq 2 n - \sqrt{5n \log\nicefrac{1}{q}}\,;\\    
\label{eq:lb-q14}
    Q_{q,\calR^{*}}(D_{\nicefrac{1}{4}})&\geq 2.5 n - \sqrt{5n \log\nicefrac{1}{q}}.
\end{align}
\end{lemma}

\begin{lemma}\label{lem:lb-Q-plus-delta}
For $q \in (0,1)$ and $\delta \in (0,1/600] $ we have that
\begin{equation}
\label{eq:lb-q14d}    
    Q_{q,\calR^{*}}(D_{\nicefrac{1}{4}+\delta})\geq 2.5 n + \nicefrac{2}{3}\;\delta n - 5\sqrt{n \log \nicefrac{1}{q}}.
\end{equation}
\end{lemma}
We prove Lemmas~\ref{lem:lb-Q} and \ref{lem:lb-Q-plus-delta} in Section~\ref{sec:lem:lb-Q}.
Lemma~\ref{lem:lb-Q}, combined with inequality~\eqref{cor:lem:DKW}, immediately yield the lower bounds in \eqref{eq:lb-without-hpol} and \eqref{eq:lb-with-hpol} on the total 
cost of printing random strands without and with homopolymers. 

We now show how to strengthen the lower bound for unconstrained strands, obtaining the desired inequality~\eqref{eq:strong-lb-with-hpol}. We will prove that, with probability at least  $1-\exp(- c' \sqrt{M}/k^2)$ for some absolute constant $c' > 0$, there exists a batch with cost at least 
\begin{equation}\label{eq:lb-cost-one-batch}
2.5 n + \nicefrac{2}{3}\;\delta n - 5\sqrt{n \log (6k)},
\end{equation}
where $\delta=\min\Big\{\sqrt{\frac{\log (M/16)}{16 n}}, 1/600\Big\}$. Thus, using the lower bound from the first part of the theorem for all the other $k-1$ batches, 
with probability at least 
\[
1-\exp(- c' \sqrt{M}/k^2) - 2\exp(- M / (4k^{2}) ) 
\geq 
1-\exp(- c_{3} \sqrt{M}/k^2),
\]
the total cost of printing $\calS_{D_{\nicefrac{1}{4}}, M}$ using $k$ batches is at least 
\begin{multline*}
\bigg(2.5 n + \nicefrac{2}{3}\;\delta n - 5 \sqrt{n \log (6k)}\bigg) + (k-1)\cdot \bigg(2.5 n -\sqrt{5 n \log (2k)}\bigg) \\
\begin{aligned}
&\geq 
k \cdot 2.5 n 
+ \nicefrac{2}{3}\;\delta n 
- 5 k \sqrt{n \log(6k)} \\
&\geq 
k \cdot 2.5 n 
+ \sqrt{n} \left[ \frac{1}{900} \min \left\{ \sqrt{\log (M/16)}, \sqrt{n} \right\} - 5 k \sqrt{\log(6k)} \right].
\end{aligned}
\end{multline*}
Now if 
$k\leq c_1 \min\{\sqrt{\log M/\log\log M},\sqrt{n/\log n}\}$ 
for a small enough $c_{1} > 0$, 
then the quantity in the display above is bounded from below by 
$k \cdot 2.5 n + c_{2} \sqrt{n \min \left\{ \log M, n \right\}}$ 
for some $c_{2} > 0$. 
Thus we have obtained~\eqref{eq:strong-lb-with-hpol}.

What remains is to prove the claim in~\eqref{eq:lb-cost-one-batch}. 
The proof of this relies on the following lemma, the proof of which we defer to Section~\ref{sec:lem:S-partition}.

\begin{lemma}\label{lem:S-partition}
For every $0 \leq \delta \leq \min\Big\{\sqrt{\frac{\log (M/16)}{16 n}}, 0.1\Big\}$ there exists a coupling  $(\calS,\calS')$ of random multisets containing strands of length $n$ that satisfies the following: \begin{enumerate}[(a)]
\item $\calS$ has the same distribution as $\calS_{D_{\nicefrac{1}{4}}, M}$;
\item $\calS'$ has the same distribution as $\calS_{D_{\nicefrac{1}{4}+\delta}, \floor{\sqrt{M}}}$; and 
\item there exists an absolute constant $c>0$ such that 
\[
\p\{|\calS \cap \calS'|\geq |\calS'|/3\}\geq 
1 - \exp(-c\sqrt{M}).
\]
\end{enumerate}
\end{lemma}
Consider the pair of random multisets $(\calS, \calS')$ from Lemma~\ref{lem:S-partition}. Since $\calS$ has the same distribution as $\calS_{D_{\nicefrac{1}{4}},M}$, it suffices to show that, with probability at least  $1-\exp(- c' \sqrt{M}/k^{2})$, in every partitioning of $\calS$ into $k$ batches there exists at least one batch with cost at least $2.5 n + \nicefrac{2}{3}\;\delta n - 5\sqrt{n \log (6k)}$. 
Let $\calS'' := \calS\cap \calS'$. Since $\calS''\subseteq \calS$, one of the batches must contain at least a $\nicefrac{1}{k}$ fraction of all strands in $\calS''$. Denote this batch by $\calB$ and let $\calB'' := \calB\cap\calS''$. If $|\calS''|\geq |\calS'|/3$, then 
$|\calB''|\geq |\calS''|/k\geq |\calS'|/(3k)$. Consequently, the cost of printing $\calB''$ is at least $\EQ_{\nicefrac{1}{(3k)},\calR^*}(\calS')$. Since $\calB''\subseteq \calB$, the cost of printing $\calB$ is also at least $\EQ_{\nicefrac{1}{(3k)},\calR^*}(\calS')$. 
By Lemma~\ref{lem:S-partition}, we have that $|\calS''|\geq |\calS'|/3$ with probability at least  $1-\exp(- c \sqrt{M})$. Thus, with this probability, the cost of printing $\calB$ is at least $\EQ_{\nicefrac{1}{(3k)},\calR^*}(\calS')$. 

Now by the inequality~\eqref{cor:lem:DKW} (replacing $\nicefrac{1}{k}$ with $\nicefrac{1}{(3k)}$ and $M$ with $\floor{\sqrt{M}}$) we have that 
\[
\p\Big\{\EQ_{\nicefrac{1}{(3k)},\calR^*}(\calS')\geq Q_{\nicefrac{1}{(6k)},\calR^*}(D_{\nicefrac{1}{4}+\delta})\Big\} 
\geq 1 - 2|\calR^*| e^{-\sqrt{M}/(20k^{2})} 
> 1 - 2 e^{-\sqrt{M}/(40k^{2})},  
\]
where in the second inequality we used that 
$|\calR^{*}| < \exp(6n) \leq \exp( \sqrt{M} / (40k^{2}))$; 
this holds due to the assumption that 
$k \leq \sqrt[4]{M}/\sqrt{240n}$. 
Finally, by Lemma~\ref{lem:lb-Q-plus-delta} we have that 
\[
Q_{\nicefrac{1}{(6k)},\calR^*}(D_{\nicefrac{1}{4}+\delta}) 
\geq 
2.5n + \nicefrac{2}{3}\,\delta n - 5\sqrt{n \log (6k)}.
\]
Putting everything together we obtain that, 
with probability at least 
$1 - 2 \exp(-\sqrt{M}/(40k^{2})) - \exp( - c\sqrt{M})$, 
the cost of printing $\calB$ is at least 
$2.5n + \nicefrac{2}{3}\,\delta n - 5\sqrt{n \log (6k)}$. 
This concludes the proof of the claim in~\eqref{eq:lb-cost-one-batch}.

We complete the proof of Theorem~\ref{thm:lb-4-homoplymers} in Sections~\ref{sec:lem:lb-Q} and~\ref{sec:lem:S-partition}, where we prove Lemmas~\ref{lem:lb-Q},~\ref{lem:lb-Q-plus-delta}, and~\ref{lem:S-partition}.
\end{proof}

\subsection{Proofs of Lemma~\ref{lem:lb-Q} and Lemma~\ref{lem:lb-Q-plus-delta}}\label{sec:lem:lb-Q}

\begin{proof}[Proof of Lemma~\ref{lem:lb-Q}]
We show that 
$Q_{q,\calR^{*}}(D_{0})\geq 2n - \Delta_{n,q}$ and
$Q_{q,\calR^{*}}(D_{\nicefrac{1}{4}})\geq 2.5n - \Delta_{n,q}$, where $\Delta_{n,q} := \sqrt{5n \log \nicefrac{1}{q}}$. To establish this inequality, it suffices to prove that for all $R\in \calR^{*}$, 
\begin{align*}
\p_{S\sim D_{0}}\big\{\cost_{R} (S) \leq 2n - \Delta_{n,q} \big\} &< q;\\
\p_{S\sim D_{\nicefrac{1}{4}}}\big\{\cost_{R} (S) \leq 2.5n - \Delta_{n,q} \big\} &< q.    
\end{align*}
To do this, we use the stochastic domination argument that we previously used in Section~\ref{sec:single_batch_LB_proofs}. Let $Y_1,\ldots,Y_n$ be i.i.d.\ random variables that are uniformly distributed on $\{1,2,3\}$, and let 
$Z_1,\ldots,Z_n$ be i.i.d.\ random variables that are uniformly distributed on $\{1,2,3,4\}$. Then, by the same arguments as in Section~\ref{sec:single_batch_LB_proofs}, we have that  
\begin{align*}
\p_{S\sim D_{0}}\big\{\cost_{R} (S) \leq 2n - \Delta_{n,q} \big\} &\leq \p\bigg\{\sum_{i=1}^n Y_i \leq 2n - \Delta_{n,q} \bigg\};\\
\p_{S\sim D_{\nicefrac{1}{4}}}\big\{\cost_{R} (S) \leq 2.5n - \Delta_{n,q} \big\} &\leq \p\bigg\{\sum_{i=1}^n Z_i \leq 2.5n - \Delta_{n,q} \bigg\}.    
\end{align*}
Note that $\E[Y_1] = 2$ and $\E[Z_1]=2.5$. Thus  
by Hoeffding's inequality (Theorem~\ref{thm:hoeffding}) we have that 
\begin{align*}
\p\bigg\{\sum_{i=1}^n Y_i \leq 2n - \Delta_{n,q} \bigg\}& \leq\exp\bigg(-\frac{2\Delta^2_{n,q}}{4n}\bigg)<q;\\
\p\bigg\{\sum_{i=1}^n Z_i \leq 2.5n - \Delta_{n,q} \bigg\}&\leq\exp\bigg(-\frac{2\Delta^2_{n,q}}{9n}\bigg)<q.    
\end{align*}
This completes the proof of \eqref{eq:lb-q0} and  \eqref{eq:lb-q14}. 
\end{proof}

\begin{proof}[Proof of Lemma~\ref{lem:lb-Q-plus-delta}] 
To prove the claim it suffices to show that for all $R \in \calR^{*}$ we have that 
\begin{equation}\label{eq:lb-q14d-suff}
\p_{S \sim D_{\nicefrac{1}{4}+\delta}} \left\{ \cost_{R}(S) \leq 2.5 n + \nicefrac{2}{3}\;\delta n - 5\sqrt{n \log \nicefrac{1}{q}} \right\} < q.
\end{equation}
Accordingly, we fix $R \in \calR^{*}$ for the rest of the proof and show~\eqref{eq:lb-q14d-suff}. 

As in Section~\ref{sec:single_batch_LB_proofs}, for $i \geq 1$ let $\tau_{i}(S,R)$ denote the index of the base of $R$ that is used to print the $i$th base of $S$, and let $\tau_{0}(S,R) = 0$ for notational convenience. 
For $i \geq 1$, define $X_{i}(S,R) := \tau_{i} (S,R) - \tau_{i-1}(S,R)$, and also let $Y_{i} := \min \{ X_{i}(S,R), 5 \}$. 
When $S \sim D_{\nicefrac{1}{4}}$, we have seen (see Section~\ref{sec:single_batch_LB_proofs}) that the distribution of $Y_{i}$ stochastically dominates the uniform distribution on $\{1,2,3,4\}$, and in particular 
%for $i \geq 1$ 
we have that 
$\E_{S \sim D_{\nicefrac{1}{4}}} \left[ Y_{i} \, \middle| \, \tau_{i-1}(S,R) \right] 
\geq 2.5$. 
When $S \sim D_{\nicefrac{1}{4} + \delta}$, 
this inequality does not necessarily hold for all $i \geq 1$. 
However, we still have the following. 
\begin{claim}\label{cl:eq:-2delta}
For all $i \geq 1$ we have that 
\begin{equation}\label{eq:-2delta}
\E_{S \sim D_{\nicefrac{1}{4}+\delta}} [Y_i\mid \tau_{i-1}(S, R)]\geq 2.5 -2\delta.
\end{equation}
\end{claim}
\begin{proof} 
Given $\tau_{i-1}(S,R)$ (and the knowledge of the fixed reference strand $R$), we know the \emph{support} of the random variable $X_{i}(S,R)$, which consists of four distinct positive integers. We also know that the conditional probabilities of taking on each particular value are given by 
$\nicefrac{1}{4} + \delta$, 
$\nicefrac{1}{4} - \nicefrac{\delta}{3}$, 
$\nicefrac{1}{4} - \nicefrac{\delta}{3}$, 
and $\nicefrac{1}{4} - \nicefrac{\delta}{3}$, 
with some particular permutation. 
The conditional expectation is thus minimized when the support is $\{1,2,3,4\}$ and the largest probability $\nicefrac{1}{4} + \delta$ is assigned to $1$. Hence,  
\[
\E_{S \sim D_{\nicefrac{1}{4}+\delta}} [Y_i\mid \tau_{i-1} (S, R)] \geq (\nicefrac{1}{4}+\delta)\cdot 1
+(\nicefrac{1}{4}-\nicefrac{\delta}{3})\cdot 2
+(\nicefrac{1}{4}-\nicefrac{\delta}{3})\cdot 3
+(\nicefrac{1}{4}-\nicefrac{\delta}{3})\cdot 4
= 2.5 -2\delta. \qedhere
\]
\end{proof}

We now show that, by averaging over three consecutive terms, we can obtain a better lower bound on the conditional expectation that is strictly greater than $2.5$ on average.
\begin{lemma}\label{lem:expect-biased-distrib}
For every $i \in \{ 1, \ldots, n-2 \}$ and $\delta \in (0,1/600]$, we have that 
\begin{equation}\label{eq:exp_geq_avg}
\E_{S\sim D_{\nicefrac{1}{4}+\delta}}\left[\frac{Y_i+Y_{i+1}+Y_{i+2}}{3}
\, \middle| \, \tau_{i-1}(S, R) \right]\geq 2.5+\nicefrac{2}{3}\,\delta.
\end{equation}
\end{lemma}
\begin{proof}
For $i \geq 1$ let $\tau^{*} := \tau_{i-1} (S, R)$ 
and note that for $i \geq 2$, by definition, $R_{\tau^{*}}$ is the base of the reference strand $R$ which is used for printing base $i-1$ of strand $S$. 
Consider the next $12$ bases of the reference strand, that is, the substrand $R_{\tau^* + 1}, \cdots, R_{\tau^* + 12}$ of $R$. We examine two cases.

First, suppose that this substrand is a triple repetition of some permutation of the bases $A$, $C$, $G$, $T$. For example,
$$R_{\tau^* + 1}, \cdots, R_{\tau^* + 12}=ACGT\,ACGT\,ACGT.$$
Then, for $j\in \{i+1,i+2\}$, the conditional distribution of $Y_{j}$ given $\tau_{i-1}(S,R)$ is given by: $Y_j=1$, $Y_j=2$, and $Y_j=3$ all with probability $\nicefrac{1}{4}-\nicefrac{\delta}{3}$, 
and $Y_j=4$ with probability $\nicefrac{1}{4} + \delta$. Thus,
$$\E_{S\sim D_{\nicefrac{1}{4}+\delta}}\left[Y_{i+1}+Y_{i+2}
\, \middle| \, \tau_{i-1} (S, R)\right] = 2\Big( (\nicefrac{1}{4}-\nicefrac{\delta}{3})\cdot (1+2+3)+
(\nicefrac{1}{4}+\delta)\cdot 4\Big) = 2(2.5+2\delta).$$
By Claim~\ref{cl:eq:-2delta} we have that 
$\E_{S\sim D_{\nicefrac{1}{4}+\delta}}\big[Y_{i}
\mid \tau_{i-1} (S, R)\big] \geq 2.5 - 2\delta$. 
By adding up these bounds, we obtain the desired bound~\eqref{eq:exp_geq_avg}.

\medskip

Now consider the second case: that the substrand $R_{\tau^* + 1}, \cdots, R_{\tau^* + 12}$ is not a triple repetition of the same permutation of $A$, $C$, $G$, $T$. Then, this substrand must contain four consecutive bases that are not a permutation of $A$, $C$, $G$, $T$. Let $\tau'$ be the index of the base just before the first such quadruple of bases. By construction we must have $\tau' \in \{ \tau^{*}, \tau^{*} + 1, \ldots, \tau^{*} + 8 \}$. 
Since $R_{\tau'+1},R_{\tau'+2},R_{\tau'+3},R_{\tau'+4}$ is not a permutation of 
$A$, $C$, $G$, $T$, at least one letter from the alphabet $\{A,C,G,T\}$ is absent in this quadruple. Hence, if $\tau_{j}(S, R)=\tau'$ for some $j$, then the expected next jump is at least 
\begin{align}\label{eq:extra-0.25}
\E_{S\sim D_{\nicefrac{1}{4}+\delta}} [Y_{j+1}\mid \tau_j(S, R)=\tau'] &\geq 
(\nicefrac{1}{4}+\delta)\cdot 1 + 
(\nicefrac{1}{4}-\nicefrac{\delta}{3})\cdot 2 +
(\nicefrac{1}{4}-\nicefrac{\delta}{3})\cdot 3 +
(\nicefrac{1}{4}-\nicefrac{\delta}{3})\cdot 5\notag\\
&= 2.5 + (0.25 -\nicefrac{7}{3}\,\delta).
\end{align}

We now distinguish three cases depending on whether 
$\tau' = \tau^{*}$ 
or 
$\tau' \in \left\{ \tau^{*} + 1, \ldots, \tau^{*} + 4 \right\}$ 
or 
$\tau' \in \left\{ \tau^{*} + 5, \ldots, \tau^{*} + 8 \right\}$. 
First, if $\tau' = \tau^{*}$, 
then by~\eqref{eq:extra-0.25} we have that 
$\E_{S\sim D_{\nicefrac{1}{4}+\delta}} [Y_{i} \mid \tau_{i-1} (S, R)] 
\geq 2.5 + (0.25 -\nicefrac{7}{3}\,\delta)$.  
We also have, 
by Claim~\ref{cl:eq:-2delta} and the tower rule, 
that 
$\E_{S \sim D_{\nicefrac{1}{4}+\delta}} [Y_{j} \mid \tau_{i-1}(S, R)]\geq 2.5 -2\delta$ 
for $j \in \{ i+1, i+2 \}$. 
Putting these bounds together we have, if $\tau' = \tau^{*}$, that 
\begin{align*}
\E_{S\sim D_{\nicefrac{1}{4}+\delta}}\left[ Y_{i} + Y_{i+1} + Y_{i+2}
\, \middle| \, \tau_{i-1} (S, R)\right] 
&\geq 
2.5 + (0.25 -\nicefrac{7}{3}\,\delta) 
+ 2 ( 2.5 -2\delta ) \\
&= 3 \cdot 2.5 + 0.25 - \nicefrac{19}{3}\,\delta 
\geq 
3 ( 2.5 + \nicefrac{2}{3}\,\delta), 
\end{align*}
where the last inequality holds when $\delta \leq 0.03$. This shows~\eqref{eq:exp_geq_avg} in this case. 

Next, suppose that $\tau' \in \left\{ \tau^{*} + 1, \ldots, \tau^{*} + 4 \right\}$. 
By Claim~\ref{cl:eq:-2delta} we have that 
\begin{equation}\label{eq:tau'_2_1}
\E_{S \sim D_{\nicefrac{1}{4}+\delta}} [ Y_{i} + Y_{i+2} \mid \tau_{i-1}(S, R)]\geq 2(2.5 -2\delta).
\end{equation}
In order to bound 
$\E_{S \sim D_{\nicefrac{1}{4}+\delta}} [Y_{i+1} \mid \tau_{i-1}(S, R)]$, 
we first condition on $\tau_{i}(S,R)$. 
Note that since $\tau' > \tau^{*}$, the bases 
$R_{\tau^{*}+1}, \ldots, R_{\tau^{*}+4}$ are all different. 
Therefore we must have that $\tau_{i}(S,R) \in \{ \tau^{*}+1, \ldots, \tau^{*}+4\}$. 
Conditioning on the value of $\tau_{i}(S,R)$ we thus have that
\begin{align*}
\E [Y_{i+1} \mid \tau_{i-1}(S, R)] 
&= \sum_{\ell = 1}^{4} \E [Y_{i+1} \mid \tau_{i}(S, R) = \tau^{*} + \ell ] 
\p \left\{ \tau_{i}(S,R) = \tau^{*} + \ell \, \middle| \, \tau_{i-1}(S,R) \right\} \\
&= 
\E [Y_{i+1} \mid \tau_{i}(S, R) = \tau' ] 
\p \left\{ \tau_{i}(S,R) = \tau' \, \middle| \, \tau_{i-1}(S,R) \right\} \\ 
&\quad 
+ \E [Y_{i+1} \mid \tau_{i}(S, R) \neq \tau' ] \left( 1 - 
\p \left\{ \tau_{i}(S,R) = \tau' \, \middle| \, \tau_{i-1}(S,R) \right\} \right), 
\end{align*}
where all expectations and probabilities are under $S \sim D_{\nicefrac{1}{4} + \delta}$. 
We then use~\eqref{eq:extra-0.25} to bound from below the first expectation in the sum above, and we use Claim~\ref{cl:eq:-2delta} to bound from below the second expectation. 
Plugging in these bounds we obtain that 
\begin{align}
\E_{S \sim D_{\nicefrac{1}{4} + \delta}} [Y_{i+1} \mid \tau_{i-1}(S, R)] 
&\geq 2.5 - 2 \delta + (0.25 - \nicefrac{\delta}{3}) \p_{S \sim D_{\nicefrac{1}{4} + \delta}} \left\{ \tau_{i}(S,R) = \tau' \, \middle| \, \tau_{i-1}(S,R) \right\} \notag \\
&\geq 2.5 - 2 \delta + (0.25 - \nicefrac{\delta}{3})^{2} 
\geq 2.5 + 6 \delta, \label{eq:tau'_2_2}
\end{align}
where the last inequality holds when $\delta \leq 1/200$. Putting together~\eqref{eq:tau'_2_1} and~\eqref{eq:tau'_2_2} we obtain~\eqref{eq:exp_geq_avg} in this case as well. 

Finally, suppose that $\tau' \in \left\{ \tau^{*} + 5, \ldots, \tau^{*} + 8 \right\}$. 
By Claim~\ref{cl:eq:-2delta} we have that 
\begin{equation}\label{eq:tau'_3_1}
\E_{S \sim D_{\nicefrac{1}{4}+\delta}} [ Y_{i} + Y_{i+1} \mid \tau_{i-1}(S, R)]\geq 2(2.5 -2\delta).
\end{equation} 
In order to bound 
$\E_{S \sim D_{\nicefrac{1}{4}+\delta}} [Y_{i+2} \mid \tau_{i-1}(S, R)]$, 
we first condition on $\tau_{i+1}(S,R)$.  
Note that since $\tau' > \tau^{*} + 4$, 
the bases $R_{\tau^{*}+1}, \ldots, R_{\tau^{*}+4}$ are all different, 
and the bases $R_{\tau^{*}+5}, \ldots, R_{\tau^{*}+8}$ 
are a repetition of the bases $R_{\tau^{*}+1}, \ldots, R_{\tau^{*}+4}$. 
This implies that 
$\tau_{i}(S,R) \in \{ \tau^{*}+1, \ldots, \tau^{*}+4\}$, 
and, moreover, that given $\tau_{i}(S,R)$, 
we have that 
$\tau_{i+1}(S,R) \in \{ \tau_{i}(S,R)+1, \ldots, \tau_{i}(S,R)+4\}$. 
In particular, note that 
if $\tau_{i}(S,R) = \tau' - 4$ and $S_{i+1} = S_{i}$, then $\tau_{i+1}(S,R) = \tau'$. 
Thus conditioning on whether or not 
$\tau_{i+1}(S,R) = \tau'$ occurs, we have that 
\begin{align*}
\E [Y_{i+2} \mid \tau_{i-1}(S, R)] 
&= 
\E [Y_{i+2} \mid \tau_{i+1}(S, R) = \tau' ] 
\p \left\{ \tau_{i+1}(S,R) = \tau' \, \middle| \, \tau_{i-1}(S,R) \right\} \\ 
&\quad 
+ \E [Y_{i+2} \mid \tau_{i+1}(S, R) \neq \tau' ] \left( 1 - 
\p \left\{ \tau_{i+1}(S,R) = \tau' \, \middle| \, \tau_{i-1}(S,R) \right\} \right). 
\end{align*}
As before, 
we use~\eqref{eq:extra-0.25} to bound from below the first expectation in the sum above, and we use Claim~\ref{cl:eq:-2delta} to bound from below the second expectation. Plugging in these bounds we obtain that 
\begin{equation*}
\E_{S \sim D_{\nicefrac{1}{4} + \delta}} [Y_{i+2} \mid \tau_{i-1}(S, R)] 
\geq 2.5 - 2 \delta + (0.25 - \nicefrac{\delta}{3}) \p_{S \sim D_{\nicefrac{1}{4} + \delta}} \left\{ \tau_{i+1}(S,R) = \tau' \, \middle| \, \tau_{i-1}(S,R) \right\}. 
\end{equation*} 
Since 
$\tau_{i}(S,R) = \tau' - 4$ and $S_{i+1} = S_{i}$ 
together imply that $\tau_{i+1}(S,R) = \tau'$, 
we have that 
%we may bound from below the probability as follows:
\begin{multline*}
\p_{S \sim D_{\nicefrac{1}{4} + \delta}} \left\{ \tau_{i+1}(S,R) = \tau' \, \middle| \, \tau_{i-1}(S,R) \right\} \\
\geq 
\p_{S \sim D_{\nicefrac{1}{4} + \delta}} \left\{ \tau_{i}(S,R) = \tau' - 4 \text{ and } S_{i+1} = S_{i} \, \middle| \, \tau_{i-1}(S,R) \right\} 
\geq \left( \nicefrac{1}{4} - \nicefrac{\delta}{3} \right) \left( \nicefrac{1}{4} + \delta \right) > 1/16.
\end{multline*}
Putting together the two previous displays we obtain that 
\begin{equation}\label{eq:tau'_3_2}
\E_{S \sim D_{\nicefrac{1}{4}+\delta}} [ Y_{i+2} \mid \tau_{i-1}(S, R)] 
\geq 
2.5 + \nicefrac{1}{64} - ( 2 + \nicefrac{1}{48}) \delta 
\geq 2.5 + 6 \delta, 
\end{equation} 
where the last inequality holds whenever $\delta \leq 1/600$. 
Putting together~\eqref{eq:tau'_3_1} and~\eqref{eq:tau'_3_2} we obtain~\eqref{eq:exp_geq_avg} in this case as well. This concludes all cases and hence concludes the proof. 
\end{proof}

We now conclude the proof of~\eqref{eq:lb-q14d-suff}. For simplicity, we assume in the following that $n$ is divisible by $3$; the case of $n$ not being divisible by $3$ can be handled by padding. 
Let $Z_{0} := 0$ and for $t \in \{1, \ldots, n/3 \}$ let 
$$Z_t := \sum_{i=1}^{3t}Y_i - 3(2.5+\nicefrac{2}{3}\,\delta)t.$$ 
By definition we have that 
$$\cost_{R} (S) = \sum_{i=1}^n X_i(S, R) \geq \sum_{i=1}^{n} Y_i = Z_{n/3} + (2.5 + \nicefrac{2}{3}\,\delta) n.$$ 
and so 
\begin{equation}\label{eq:bounding_via_submg}
\p_{S \sim D_{\nicefrac{1}{4}+\delta}} \left\{ \cost_{R}(S) \leq 2.5 n + \nicefrac{2}{3}\;\delta n - 5\sqrt{n \log \nicefrac{1}{q}} \right\} 
\leq 
\p_{S \sim D_{\nicefrac{1}{4}+\delta}} \left\{ Z_{n/3} \leq - 5 \sqrt{n \log \nicefrac{1}{q}} \right\}.
\end{equation}
Lemma~\ref{lem:expect-biased-distrib} implies that 
$\left\{ Z_{t} \right\}_{t=0}^{n/3}$ 
is a submartingale with respect to its natural filtration. 
Since $Y_i\in\{1,\dots,5\}$ for every $i \geq 1$, 
the submartingale differences satisfy 
$Z_{t+1} - Z_{t} 
= (Y_{3t+1} + Y_{3t+2} + Y_{3t+3}) - 3(2.5+\nicefrac{2}{3}\delta) \in [-4.5-2\delta,7.5-2\delta]$, 
that is, they take values in an interval of length $12$. 
Thus by the Azuma-Hoeffding inequality (see, e.g.,~\cite{boucheron2013concentration}) we have that 
\begin{equation}\label{eq:azuma}
\p_{S \sim D_{\nicefrac{1}{4}+\delta}} \left\{ Z_{n/3} \leq - 5 \sqrt{n \log \nicefrac{1}{q}} \right\} 
\leq \exp \left( - \frac{2 \cdot 5^{2} n \log \nicefrac{1}{q}}{(n/3) \cdot 12^{2}} \right) 
= q^{25/24} 
< q.
\end{equation}
Putting together~\eqref{eq:bounding_via_submg} and~\eqref{eq:azuma} we obtain~\eqref{eq:lb-q14d-suff}. 
\end{proof}

\subsection{Proof of Lemma~\ref{lem:S-partition}}\label{sec:lem:S-partition}
\begin{proof}[Proof of Lemma~\ref{lem:S-partition}]
For simplicity of exposition, assume that $\sqrt{M}$ is an integer number (if it is not, we can replace $M$ with $M'=\floor{\sqrt{M}}^2$). Let $\Sigma := \{ A, C, G, T \}$ and denote by $\cL$ the set of all strands $S$ of length $n$ with the number of repeating bases, $d(S)$, being at most  
$\ceil{(\nicefrac{1}{4}+\delta)n}$:
$$\cL = \{S\in \Sigma^n: d(S) \leq \ceil{(\nicefrac{1}{4}+\delta)n} \}.$$
Let $\phi := 1/\sqrt{M}$. We first show the following claim. 
\begin{claim}\label{clm:ratio-D} 
Suppose that  
$0 \leq \delta \leq \min \left\{ \sqrt{\frac{\log (M/16)}{16 n}}, 0.1 \right\}$. 
For all strands $S\in \cL$ we have that 
\begin{equation}\label{eq:D-prob-ratio}
D_{\nicefrac{1}{4}}(S)\geq 2 \phi \cdot D_{\nicefrac{1}{4}+\delta}(S).
\end{equation}
\end{claim}
\begin{proof}
First observe that for all $S \in \Sigma^{n}$ we have that 
$$
\frac{D_{\nicefrac{1}{4}+\delta}(S)}{D_{\nicefrac{1}{4}}(S)} = 
\frac{\frac{1}{4}\cdot \big(\frac{1}{4} - \frac{\delta}{3}\big)^{n-d(S)-1}\big(\frac{1}{4} + \delta\big)^{d(S)}}{\frac{1}{4^n}} = 
\big(1 - \nicefrac{4}{3}\,\delta\big)^{n-d(S)-1}\big(1 + 4\delta\big)^{d(S)}.
$$
For $S\in \cL$ we have that $d(S)\leq (\nicefrac{1}{4}+\delta)n+1$ by definition. Plugging this bound into the display above we obtain for $S \in \cL$ that 
\begin{align*}
\frac{D_{\nicefrac{1}{4}+\delta}(S)}{D_{\nicefrac{1}{4}}(S)} 
&\leq 
\big(1 - \nicefrac{4}{3}\,\delta\big)^{(\nicefrac{3}{4}-\delta)n-2}\big(1 + 4\delta\big)^{(\nicefrac{1}{4}+\delta)n+1} \\ 
&=
\Big((1-\nicefrac{4}{3}\;\delta)^3(1+4\delta)\Big)^{n/4} \cdot
\left(\frac{1+4\delta}{1-\nicefrac{4}{3}\;\delta}\right)^{\delta n}\cdot \left(\frac{(1+4\delta)}{(1-\nicefrac{4}{3}\;\delta)^2}\right).
\end{align*}
We now bound from above the three factors on the right hand side. The first factor is at most $1$, since $(1-\nicefrac{4}{3}\;\delta)^3 (1+4\delta) \leq 1$ by the AM-GM inequality. 
The third factor is bounded from above by $2$ for $\delta\in[0,\nicefrac{1}{10}]$. 
Finally, to bound the second factor, observe that the quadratic function
$$(1+4x) - (1-4x/3)(1+8x) = \frac{32\, x(x-\nicefrac{1}{4})}{3}$$
is negative for $x\in (0,\nicefrac{1}{4})$. Hence, $(1+4\delta)/(1-\nicefrac{4}{3}\;\delta)\leq (1+8\delta)$ 
for $\delta \in [0,1/4]$ and so 
$$\left(\frac{1+4\delta}{1-\nicefrac{4}{3}\;\delta}\right)^{\delta n}
\leq (1+8\delta)^{\delta n} = 
e^{\delta n\ln(1+8\delta)}\leq 
e^{8\delta^2 n}.$$
Combining these bounds we obtain that
$$D_{\nicefrac{1}{4}}(S)\geq \nicefrac{1}{2}\;e^{-8\delta^2 n}
\cdot D_{\nicefrac{1}{4}+\delta}(S) 
\geq 2 \phi \cdot D_{\nicefrac{1}{4}+\delta}(S),$$
where in the last inequality we used that 
$0 \leq \delta \leq \sqrt{\frac{\log (M/16)}{16 n}}$ 
and that 
$\phi = 1 / \sqrt{M}$. 
\end{proof}
We also need to estimate the probability of the set $\cL$.
\begin{claim}\label{clm:measure-L}
We have that $D_{\nicefrac{1}{4}+\delta}(\cL)\geq \nicefrac{1}{2}$.
\end{claim}
\begin{proof}
The number of repetitions $d(S)$ for random strands $S\sim D_{p}$ has a binomial distribution $B(n-1,p)$. The median of the binomial distribution $B(n-1,p)$ is at most $\ceil{(n-1)p}$. Thus, $\p_{S\sim D_p}\{d(S)\leq \ceil{(n-1)p}\} \geq 1/2$ 
and hence $D_{\nicefrac{1}{4}+\delta}(\cL)\geq \nicefrac{1}{2}$.
\end{proof}
We now return to the proof of Lemma~\ref{eq:D-prob-ratio}.
Define two probability measures $A$ and $B$ on $\Sigma^n$:
$$
A(S) := 
\begin{cases}
\frac{D_{\nicefrac{1}{4}+\delta}(S)}{D_{\nicefrac{1}{4}+\delta}(\cL)}, & \text{ for } S\in\cL;\\
0, & \text{ for } S\notin\cL;
\end{cases}\;\;\;\;\;\;\;\;
B(S) :=\frac{D_{\nicefrac{1}{4}}(S) - \phi \, A(S)}{1-\phi}.
$$
\begin{claim}~\label{clm:A-B-prob-measure}
$A$ and $B$ are probability measures on $\Sigma^n$. That is,
$A$ and $B$ are nonnegative and $A(\Sigma^n) = B(\Sigma^n) = 1$.
\end{claim}
\begin{proof} 
If $S\sim D_{\nicefrac{1}{4}+\delta}$, then the conditional distribution of $S$ given that $S \in \cL$ is exactly given by~$A$. In other words, for every $S' \in \Sigma^{n}$ we have that 
$\p_{S\sim D_{\nicefrac{1}{4}+\delta}} ( S = S' \mid S \in \cL ) = A(S')$. 
Therefore $A$ is a probability measure. 
It is easy to verify that $B(\Sigma^n) = 1$:
$$B(\Sigma^n) = \frac{D_{\nicefrac{1}{4}}(\Sigma^n) - \phi \, A(\Sigma^n)}{1-\phi} = \frac{1-\phi}{1-\phi} = 1.$$
Thus, it remains to show that $B$ is nonnegative. For $S \notin \cL$ we have that $A(S) = 0$, so $B(S) = D_{\nicefrac{1}{4}}(S) / (1 - \phi) \geq 0$. 
By Claim~\ref{clm:measure-L} we know that 
$D_{\nicefrac{1}{4}+\delta}(\cL)\geq \nicefrac{1}{2}$, 
so for every $S \in \cL$ we have that 
$A(S) \leq 2 D_{\nicefrac{1}{4}+\delta}(S)$. 
Therefore for every $S \in \cL$ we have that 
\[
B(S) 
= \frac{D_{\nicefrac{1}{4}}(S) - \phi \, A(S)}{1-\phi} 
\geq \frac{D_{\nicefrac{1}{4}}(S) - 2 \phi \, D_{\nicefrac{1}{4}+\delta}(S)}{1-\phi} 
\geq 0,
\]
where the last inequality is a consequence of Claim~\ref{clm:ratio-D}. 
\end{proof}

%%% Main part of the proof of the coupling lemma
We now show how to define the coupling $(\calS,\calS')$ satisfying the conditions (a), (b), and (c) of Lemma~\ref{lem:S-partition}. 
Let $S_1, S_2,\dots$, be an infinite sequence of i.i.d.\ random strands, where each strand $S_i$ is distributed according to the distribution $D_{\nicefrac{1}{4}+\delta}$. 
Let $L_1,L_2,\ldots,$
be the subsequence of strands that belong to the 
set $\cL$. 
Note in particular that, by construction, $L_1,L_2,\ldots,$ are i.i.d.\ strands distributed according to the distribution $A$.

First, let  
$\cS' := \left\{ S_{1}, S_{2}, \ldots, S_{\sqrt{M}} \right\}$. 
By construction we have that $\cS'$ has the same distribution as $\cS_{D_{\nicefrac{1}{4}+\delta},\sqrt{M}}$, 
which shows part (a) of the claim. 

Next, we define $\calS$. 
We generate the strands of $\calS$ one by one. For each $i=1,\dots, M$, we flip a coin: 
with probability $\phi$, we add the first not yet selected strand from the sequence $\{L_j\}_{j \geq 1}$  to~$\cS$; 
with the remaining probability $(1-\phi)$, we add an independent random strand drawn from the distribution $B$ to $\calS$. 
By construction the strands in $\cS$ are i.i.d. 
Moreover, by construction, each strand is distributed according to 
$\phi A + (1-\phi) B = D_{\nicefrac{1}{4}}$, 
where the equality follows from the definition of $B$. 
So $\cS$ has the same distribution as $\cS_{D_{\nicefrac{1}{4}},M}$, 
which shows part (b) of the claim.

Finally, we verify part (c) of the claim. The intersection of the sets $\calS$ and $\calS'$ is a prefix of the sequence $L_1,L_2,\dots$. In expectation, the set $\calS$ contains  $\phi M=\sqrt{M}$ elements from the sequence $L_1,L_2,\ldots$. Thus, by the Chernoff bound, it contains the first $\sqrt{M}/2$ elements of this sequence with probability at least $1 - \exp(-c_1\sqrt{M})$.  Similarly, in expectation, the set $\calS'$ contains at least $\sqrt{M}/2$ elements from $L_1,L_2,\ldots$. Hence, by the Chernoff bound, it contains the first $\sqrt{M}/3$ elements from the sequence $L_1,L_2,\ldots$ with probability at least $1 - \exp(-c_2 \sqrt{M})$. Thus, by a union bound, % we have that 
\[
\p \left\{ \left\{ L_{1}, L_{2}, \ldots, L_{\sqrt{M}/3} \right\} \subseteq \calS\cap\calS' \right\} 
\geq 1 - \exp(-c_1\sqrt{M}) - \exp(-c_2 \sqrt{M}). \qedhere
\]
%which concludes the proof. 
\end{proof}

%% file: std_inequalities.tex
\section{Standard tail bounds}

We recall Hoeffding's inequality (see, e.g.,~\cite{boucheron2013concentration}), which we use throughout our proofs.

\begin{theorem}[Hoeffding's inequality]\label{thm:hoeffding}
Let $X_{1}, \ldots, X_{n}$ be independent random variables such that $X_{i}$ takes its values in $\left[ a_{i}, b_{i} \right]$ almost surely for all $i \leq n$. Let $S_{n} := X_{1} + \ldots + X_{n}$. Then for every $t > 0$ we have that
\[
\p \left\{ S_{n} - \E \left[ S_{n} \right] \geq t \right\}
\leq
\exp \left( - \frac{2 t^{2}}{\sum_{i=1}^{n} \left( b_{i} - a_{i} \right)^{2}} \right).
\]
\end{theorem}

%\todo{Note:} If we care about optimizing constants in the exponent, we can explore using other concentration inequalities, such as Bernstein's inequality.

We also recall the Paley-Zygmund inequality (see, e.g.,~\cite{boucheron2013concentration}).

\begin{lemma}[Paley-Zygmund inequality]\label{lem:PZ}
Let $Z$ be a nonnegative random variable and let $\theta \in [0,1]$. Then 
\[
\p \left\{ Z \geq \theta \E \left[ Z \right] \right\} 
\geq \left( 1 - \theta \right)^{2} \frac{\left( \E \left[ Z \right] \right)^{2}}{\E \left[ Z^{2} \right]}.
\]
\end{lemma}

We use the Paley-Zygmund inequality to prove a lower bound on the right tail of the sum of i.i.d.\ random variables that occur in the proofs of our lower bound results.

\begin{lemma}\label{lem:right_tail_LB}
Let $\ell \geq 2$ be a fixed positive integer. 
Let $\left\{ Y_{i} \right\}_{i=1}^{n}$ be i.i.d.\ random variables that are uniform on $\left\{ 1, 2, \ldots, \ell \right\}$. 
For $\lambda$ satisfying 
$0 < \lambda \leq \ell n / 50$ 
we have that 
\begin{equation}\label{eq:right_tail_LB_general}
\p \left\{ \sum_{i=1}^{n} \left( Y_{i} - \E \left[ Y_{i} \right] \right) \geq \lambda \right\} 
\geq 
\left( \exp \left( \frac{12.5 \lambda^{2}}{\ell^{2} n} \right) - 1 \right)^{2} 
\exp \left( - \frac{400 \lambda^{2}}{\ell^{2} n} \right).
\end{equation}
In particular, when 
$\ell \sqrt{n} / 3 
\leq \lambda 
\leq \ell n / 50$, 
we have that 
\begin{equation}\label{eq:right_tail_LB_specific}
\p \left\{ \sum_{i=1}^{n} \left( Y_{i} - \E \left[ Y_{i} \right] \right) \geq \lambda \right\}
\geq 
\exp \left( - \frac{400 \lambda^{2}}{\ell^{2} n} \right).
\end{equation}
\end{lemma}

We note that this bound is sharp up to a universal multiplicative constant in the exponent, as witnessed by Hoeffding's inequality (Theorem~\ref{thm:hoeffding}).

\begin{proof}
To abbreviate notation, let 
$Z_{i} := Y_{i} - \E \left[ Y_{i} \right]$ for $i \in \left[ n \right]$, 
and let 
$S_{n} := \sum_{i=1}^{n} Z_{i}$. 
For every $t > 0$ we have that 
\begin{equation}\label{eq:ineq_exp}
\p \left\{ S_{n} \geq \lambda \right\} 
= \p \left\{ \exp \left( t S_{n} \right) \geq \exp \left( t \lambda \right) \right\} 
= \p \left\{ \exp \left( t S_{n} \right) \geq \frac{\exp \left( t \lambda \right)}{\E \left[ \exp \left( t S_{n} \right) \right]} \E \left[ \exp \left( t S_{n} \right) \right] \right\}.
\end{equation}
We now apply the Paley-Zygmund inequality (Lemma~\ref{lem:PZ})  
with $Z := \exp \left( t S_{n} \right)$ 
and 
\[
\theta := \frac{\exp \left( t \lambda \right)}{\E \left[ \exp \left( t S_{n} \right) \right]}.
\]
In order to do so, we must have $\theta \in [0,1]$, so  $\lambda$ must satisfy 
\begin{equation}\label{eq:lambda_condition}
\lambda \leq \frac{1}{t} \log \E \left[ \exp \left( t S_{n} \right) \right].
\end{equation}
For $t > 0$ and $\lambda$ satisfying~\eqref{eq:lambda_condition} we have by~\eqref{eq:ineq_exp} and Lemma~\ref{lem:PZ} that 
\begin{equation}\label{eq:PZexpLB}
\p \left\{ S_{n} \geq \lambda \right\} 
\geq 
\left( 1 - \frac{\exp \left( t \lambda \right)}{\E \left[ \exp \left( t S_{n} \right) \right]} \right)^{2} 
\frac{\left( \E \left[ \exp \left( t S_{n} \right) \right] \right)^{2}}{\E \left[ \exp \left( 2 t S_{n} \right) \right]} 
= \frac{\left( \E \left[ \exp \left( t S_{n} \right) \right] - \exp \left( t \lambda \right) \right)^{2}}{\E \left[ \exp \left( 2 t S_{n} \right) \right]}. 
%=: f(t, \lambda).
\end{equation} 
In the following we analyze this expression. 
By an explicit calculation we have that 
\[
\E \left[ \exp \left( t Y_{1} \right) \right] 
= 
\frac{e^{t} \left( e^{t \ell} - 1 \right)}{\ell \left( e^{t} - 1 \right)}.
\]
Since $\E \left[ Y_{1} \right] = (\ell+1)/2$, we thus have that 
\[
\E \left[ \exp \left( t Z_{1} \right) \right] 
= \E \left[ \exp \left( t Y_{1} \right) \right] e^{-t(\ell+1)/2}
= \frac{e^{t/2} \left( e^{t \ell/2} - e^{-t \ell/2} \right)}{\ell \left( e^{t} - 1 \right)}. 
\]
By multiplying and dividing by $t$, we slightly rewrite this expression as 
\begin{equation}\label{eq:Z_momgenfn}
\E \left[ \exp \left( t Z_{1} \right) \right] 
= \frac{t e^{t/2}}{e^{t} - 1} \cdot \frac{e^{t \ell/2} - e^{-t \ell/2}}{t \ell}
\end{equation}
and bound the two factors separately. 
For the first factor we have for all $t > 0$ that 
\begin{equation}\label{eq:bounds_first_factor}
e^{-t^{2}/20} 
\leq 
\frac{t e^{t/2}}{e^{t} - 1}
\leq 1.
\end{equation}
For the second factor, note that 
\[
e^{x^{2}/25} 
\leq \frac{e^{x/2} - e^{-x/2}}{x} 
\leq e^{x^{2}/20},
\]
where the first inequality holds for all $x \in (0,1]$ and the second inequality holds for all $x > 0$. 
Thus for all $t \in ( 0, 1/\ell]$ we have that 
\begin{equation}\label{eq:bounds_second_factor}
e^{t^{2} \ell^{2} / 25} 
\leq 
\frac{e^{t \ell/2} - e^{-t \ell/2}}{t \ell}
\leq 
e^{t^{2} \ell^{2}/20}.
\end{equation}
Putting together the bounds in~\eqref{eq:bounds_first_factor} and~\eqref{eq:bounds_second_factor}, and plugging them into~\eqref{eq:Z_momgenfn} we obtain the upper bound 
\[
\E \left[ \exp \left( t Z_{1} \right) \right] 
\leq 
e^{t^{2} \ell^{2}/20},
\]
which holds for all $t > 0$. 
Similarly, we obtain the lower bound 
\[
\E \left[ \exp \left( t Z_{1} \right) \right] 
\geq 
\exp \left( t^{2} \ell^{2} / 25 - t^{2} / 20 \right) 
= \exp \left( t^{2} \frac{4\ell^{2} - 5}{100} \right) 
\geq e^{t^{2} \ell^{2} / 40},
\]
which holds for all $t \in (0,1/\ell]$; 
here in the last inequality we used that $4\ell^{2} - 5 \geq 2.5 \ell^{2}$, which holds for all $\ell \geq 2$. 
In summary, we have shown that for all $t \in (0,1/\ell]$ we have that 
\begin{equation}\label{eq:Z_bounds}
e^{t^{2} \ell^{2} / 40}
\leq 
\E \left[ \exp \left( t Z_{1} \right) \right] 
\leq 
e^{t^{2} \ell^{2}/20}.
\end{equation}
Since $\left\{ Z_{i} \right\}_{i=1}^{n}$ are i.i.d.\ 
we have that 
$\E \left[ \exp \left( t S_{n} \right) \right] 
= \left( \E \left[ \exp \left( t Z_{1} \right) \right] \right)^{n}$. 
Therefore by~\eqref{eq:Z_bounds} we have for all $t \in (0,1/\ell]$ that 
\begin{equation}\label{eq:Sn_bounds}
e^{t^{2} \ell^{2} n / 40}
\leq 
\E \left[ \exp \left( t S_{n} \right) \right] 
\leq 
e^{t^{2} \ell^{2} n / 20}.
\end{equation}
Plugging these inequalities back into~\eqref{eq:PZexpLB} we obtain the lower bound 
\begin{equation}\label{eq:LB_to_optimize}
\p \left\{ S_{n} \geq \lambda \right\} 
\geq \frac{\left( \exp \left( t^{2} \ell^{2} n / 40 \right) - \exp \left( t \lambda \right) \right)^{2}}{\exp \left( t^{2} \ell^{2} n /5\right)}.
\end{equation}
The inequality in~\eqref{eq:LB_to_optimize} holds 
whenever $t \in (0,1/\ell]$ and the inequality 
$t\lambda \leq t^{2} \ell^{2} n / 40$ 
holds. 
We now choose $t$ to be 
\[
t := \frac{50 \lambda}{\ell^{2} n},
\]
a choice which requires $\lambda > 0$. 
This choice of $t$ satisfies the inequality 
$t\lambda \leq t^{2} \ell^{2} n / 40$. 
Furthermore, we then have $t \leq 1/\ell$ if and only if $\lambda \leq \ell n / 50$. 
Plugging in this choice of $t$ into~\eqref{eq:LB_to_optimize} we obtain that \[
\p \left\{ S_{n} \geq \lambda \right\} 
\geq 
\frac{\left( \exp \left( \frac{62.5\lambda^{2}}{\ell^{2} n} \right) - \exp \left( \frac{50\lambda^{2}}{\ell^{2} n} \right)\right)^{2}}{\exp \left( \frac{500\lambda^{2}}{\ell^{2} n} \right)},
\]
from which~\eqref{eq:right_tail_LB_general} follows directly. 
Finally, note that when $\lambda \geq \ell \sqrt{n} / 3$, 
then 
$\exp \left( 12.5 \lambda^{2} / \left( \ell^{2} n \right) \right) - 1 \geq e - 1 > 1$, and~\eqref{eq:right_tail_LB_specific} follows. 
\end{proof}

%% file: quantile-bounds-apx.tex
\section{Proof of Lemma~\ref{lem:DKW}}\label{sec:proof:lem:DKW}
In this section we prove Lemma~\ref{lem:DKW} about the empirical quantiles of the distribution of the random variable $\cost_R(S)$, where $S$ is a random strand.
\begin{proof}[Proof of Lemma~\ref{lem:DKW}]
By the Dvoretzky--Kiefer--Wolfowitz inequality 
(\cite{DKW56, tightDKW90}), the empirical cumulative distribution function for $\cost_R(S)$ is very close to the cumulative distribution function
for $\cost_R(S) \leq t$ with high probability over the random choice of the set $\calS_{D,M}$. Specifically,
we have 
$$\p\Bigg\{
\bigg|\frac{|\{S\in \calS: \cost_R(S) \leq t\}|}{|\calS|} - 
\p\big\{\cost_R(S) \leq t\big\}\bigg|\leq \varepsilon
\text{ for all }t\in \R\Bigg\}\geq 1 -2e^{-2M\varepsilon^{2}}.$$
Let 
$$\cE := \Bigg\{
\bigg|\frac{|\{S\in \calS: \cost_R(S) \leq t\}|}{|\calS|} - 
\p\big\{\cost_R(S) \leq t\big\}\bigg|\leq \varepsilon
\text{ for all }t\in \R\Bigg\}.$$
We show that if $\cE$ occurs then
for all $q\in(\varepsilon,1-\varepsilon)$, we have
$$Q_{q-\varepsilon,R}(D)\leq \EQ_{q,R}(\calS_{D,M})\leq Q_{q+\varepsilon,R}(D),$$
and, therefore, (\ref{eq:DKW:R}) holds. Inequality~(\ref{eq:DKW:setR}) follows from 
(\ref{eq:DKW:R}) by a union bound over all $R$ in $\calR$.

\medskip
We first show that
$\EQ_{q,R}(\calS_{D,M}) \geq Q_{q-\varepsilon,R}(D)$ for all $q\in(\varepsilon,1)$. Consider an arbitrary $q\in(\varepsilon,1)$ and let $t^*= Q_{q-\varepsilon,R}(D) - 1$. Since $t^*< Q_{q-\varepsilon,R}(D)$, we have $\p\big\{\cost_R(S) \leq t^*\big\} < q-\varepsilon$ and 
$$
\frac{|\{S\in \calS: \cost_R(S) \leq t^*\}|}{|\calS|} \leq 
\p\big\{\cost_R(S) \leq t^*\big\} + \varepsilon
< q.$$
Hence, $\EQ_{q,R}(\calS_{D,M}) > t^*$, and since
$\EQ_{q,R}(\calS_{D,M})$ is an integer, 
$\EQ_{q,R}(\calS_{D,M}) \geq Q_{q-\varepsilon,R}(D)$.

\medskip

We now show that
$\EQ_{q,R}(\calS_{D,M}) \leq Q_{q+\varepsilon,R}(D)$ for all $q\in(0,1-\varepsilon)$. Consider an arbitrary $q\in(0,1-\varepsilon)$ and let $t^{**}= Q_{q+\varepsilon,R}(D)$. By the definition of $Q_{q,R}(D)$, we have $\p\big\{\cost_R(S) \leq t^{**}\big\} \geq q+\varepsilon$ and 
$$
\frac{|\{S\in \calS: \cost_R(S) \leq t^{**}\}|}{|\calS|} \geq 
\p\big\{\cost_R(S) \leq t^{**}\big\} - \varepsilon
\geq q.$$
Hence, 
$\EQ_{q,R}(\calS_{D,M}) \leq t^{**}=Q_{q+\varepsilon,R}(D)$.
\end{proof}